\documentclass[english]{article}

\usepackage[utf8]{inputenc}
\usepackage[T1]{fontenc}
\usepackage{geometry}

\usepackage{amsmath,amsfonts,amsthm,amssymb}
\usepackage{url}
\usepackage{algorithm}
\usepackage{algpseudocode}
\usepackage{bbm}
\usepackage{float}
\usepackage{framed}
\usepackage{enumerate}
\usepackage{color}
\usepackage[usenames,dvipsnames,svgnames,table]{xcolor}
\usepackage{array}
\usepackage[para]{footmisc}
\usepackage{epsfig}
\usepackage{xspace}
\definecolor{darkgreen}{rgb}{0,0.5,0}
\usepackage{hyperref}
\hypersetup{colorlinks=true, linkcolor=blue, urlcolor=blue, citecolor=darkgreen}

\definecolor{light-gray}{gray}{0.9}
\colorlet{shadecolor}{light-gray}

\makeatletter
\newcommand{\thickhline}{%
    \noalign {\ifnum 0=`}\fi \hrule height 1pt
    \futurelet \reserved@a \@xhline
}
\newcolumntype{"}{@{\hskip\tabcolsep\vrule width 1pt \hskip\tabcolsep}}
\makeatother
\newcolumntype{'}{@{\hskip\tabcolsep\vrule width 1pt}}
\makeatother
\newcolumntype{`}{@{\vrule width 1pt \hskip\tabcolsep}}
\makeatother

\newcommand{\plog}{\mathop\mathrm{polylog}}
\DeclareMathOperator*{\E}{\mathbb{E}}
\let\Pr\relax
\newcommand{\Pr}{\mathbb{P}}

\newcommand\Ex[1]{\mathbb{E}\left[\,#1\,\right]}

\newcommand{\Pro}[1]{\mathbb{P}\left[\,#1\,\right]}

\DeclareMathOperator{\tr}{tr}

\DeclareMathOperator{\sign}{sign}
\DeclareMathOperator{\orth}{orth}
\DeclareMathOperator{\chol}{chol}
\DeclareMathOperator{\gap}{gap}
\DeclareMathOperator{\bgap}{\ol{\gap}}
\DeclareMathOperator{\dist}{\mathcal{D}}

\newcommand{\euv}{\bv{e}_{u,v}}
\newcommand{\pairs}{\mathcal{P}}

\newcommand{\dw}{\mathcal{D}_{\bv{W}}}

\newcommand{\modelpq}{$(n, p,q)$-weighted communication model\xspace}
\newcommand{\modelpqshort}{$(n, p,q)$-weighted communication\xspace}
\newcommand{\modelg}{$G(n,p,q)$-communication model\xspace}
\newcommand{\modelgshort}{$G(n,p,q)$-communication\xspace}

\newcommand{\R}{\mathbb{R}}

\newcommand{\poly}{\mathop\mathrm{poly}}

\newif\ifArxiv
\Arxivfalse

\newcommand{\nzero}{1}
\newcommand{\wh}{\widehat}

\newcommand{\ol}{\overline}

\newcommand{\eqdef}{\mathbin{\stackrel{\rm def}{=}}}
\newcommand{\norm}[1]{\|#1\|}

\newcommand{\bs}[1]{\boldsymbol{#1}}
\newcommand{\bv}[1]{\mathbf{#1}}

\newcommand{\algoname}[1]{\textnormal{\textsc{#1}}}
\makeatletter

\DeclareMathOperator{\smin}{\sigma_{min}}
\DeclareMathOperator{\smax}{\sigma_{max}}

\newtheorem*{rep@theorem}{\rep@title}
\newcommand{\newreptheorem}[2]{%
\newenvironment{rep#1}[1]{%
 \def\rep@title{#2 \ref{##1}}%
 \begin{rep@theorem}}%
 {\end{rep@theorem}}}
\makeatother
\newreptheorem{theorem}{Theorem}
\usepackage{aliascnt}
\def\NewTheorem#1#2{%
  \newaliascnt{#1}{theorem}
  \newtheorem{#1}[#1]{#2}
  \aliascntresetthe{#1}
  \expandafter\def\csname #1autorefname\endcsname{#2}
}

 \newtheorem{theorem}{Theorem}[section]
\NewTheorem{lemma}{Lemma}
\NewTheorem{corollary}{Corollary}
\NewTheorem{conjecture}{Conjecture}
\NewTheorem{observation}{Observation}
\NewTheorem{assumption}{Assumption} 

\NewTheorem{remark}{Remark}

\NewTheorem{definition}{Definition}
\NewTheorem{proposition}{Proposition}

\NewTheorem{claim}{Claim}

\newcommand{\ripref}[1]{\hyperref[#1]{Algorithm}~\ref{#1}}

\newtheoremstyle{restate}{\topsep}{\topsep}{\itshape}{0pt}{\bfseries}{.}{5pt plus 1pt minus 1pt}{\thmname{#1}\thmnumber{ \begin{NoHyper}\ref{#3}\end{NoHyper}}}
\theoremstyle{restate}

\usepackage{authblk}

\def\blackslug{\hbox{\hskip 1pt \vrule width 4pt height 8pt
    depth 1.5pt \hskip 1pt}}
\def\QED{\quad\blackslug\lower 8.5pt\null\par}

  \usepackage{nth}
  \usepackage{intcalc}

\providecommand{\keywords}[1]{\textbf{\textit{keywords---}} #1}

\providecommand{\funding}[1]{\textbf{\textit{funding---}} #1}

\title{Eigenvector Computation and Community Detection in Asynchronous Gossip Models}

\author[1]{\large Frederik Mallmann-Trenn}
\author[1]{\large Cameron Musco}
\author[1]{\large Christopher Musco}

\affil[1]{\large MIT CSAIL, Cambridge MA, US\\
  \texttt{\{mallmann, cnmusco, cpmusco\}@mit.edu}}

\begin{document}
\maketitle

\begin{abstract} 
We give a simple distributed algorithm for computing adjacency matrix eigenvectors for the communication graph in an asynchronous gossip model. We show how to use this algorithm to give state-of-the-art asynchronous community detection algorithms when the communication graph is drawn from the well-studied
\emph{stochastic block model}. Our methods also apply to a natural alternative model of randomized communication, where nodes within a community communicate more frequently than nodes in different communities. 

Our analysis simplifies and generalizes prior work by
forging a connection between asynchronous eigenvector computation and Oja's algorithm for streaming principal component analysis. We hope that our work serves as a starting point for building further connections between the analysis of stochastic iterative methods, like Oja's algorithm, and work on asynchronous and gossip-type algorithms for distributed computation.

\end{abstract}

\vfill


\funding{This work
was supported in part by NSF Award Numbers BIO-1455983,
CCF-1461559,  CCF-0939370, and CCF-1565235. Cameron Musco was partially supported by an NSF graduate student fellowship.
}
\newpage

%
%



\section{Introduction}
\vspace{-2ex}
Motivated by the desire to process and analyze  increasingly  large networks---in particular social networks---considerable research has focused on finding efficient distributed protocols for problems like triangle counting, community detection, PageRank computation, and node centrality  estimation. 
Many  of the most popular systems for massive-scale graph processing, including Google's 
 Pregel \cite{pregel} and Apache Giraph \cite{gigraph} (used by Facebook), employ programming models based on the simulation of distributed message passing algorithms, in which each node is viewed as a processor that can send messages to its neighbors.
 
Apart from computational benefits, distributed graph processing can also be required when  privacy constraints apply: for example, EU regulations restrict the personal data  sent to countries outside of the EU \cite{eu}. Distributed algorithms avoid possibly problematic aggregation of network information, allowing each node to compute a local output based on their own neighborhood and messages received from their neighbors. 

One of the main problems of interest in network analysis is the computation of the eigenvectors of a networks' adjacency matrix (or related incidence matrices, such as the graph Laplacian). The extremal eigenvectors have many important applications---from graph partitioning and community detection \cite{ding2001min,MU05}, to embedding in graph-based machine learning \cite{belkin2002laplacian,ng2002spectral}, to measuring node centrality and computing importance scores like PageRank \cite{bonacich2007some}.

Due to their importance, there has been significant work on distributed eigenvector approximation. In \emph{synchronous} message  passing systems, it is possible to simulate the well-known power method for iterative eigenvector approximation \cite{kempe2008decentralized}. However, this algorithm requires that each node communicates synchronously with all of its neighbors in each round .

In an attempt to relax this requirement, models in which a subset of neighbors are sampled in each communication round \cite{korada2011gossip} have been studied.
However, the computation of graph eigenvectors in fully asynchronous and gossip-based message passing systems, in which nodes communicate with a single neighbor at a time in an asynchronous fashion, is not well-understood.
While a number of algorithms have been proposed, which give convergence to the true eigenvectors as the number of iterations goes to infinity, strong finite iteration approximation bounds are not known \cite{ghadban2015gossip,morral2012asynchronous}.
\vspace{-2.5ex}
\paragraph{Our contributions}
\vspace{-1ex}
In this work, we give state-of-the-art algorithms for graph eigenvector computation in asynchronous systems with randomized schedulers, including the classic gossip model \cite{boyd2006randomized,dimakis2010gossip} and population protocol model \cite{aspnes2009introduction}. We show that in these models, communication graph eigenvectors can be computed via a very simple adaption of Oja's classic iterative algorithm for principal components analysis \cite{Oja1982}. Our analysis leverages recent work studing Oja's algorithm for streaming covariance matrix eigenvector estimation \cite{allen2016first,aaronsPaper}.

By  making an explicit connection between work on streaming eigenvector estimation and asynchronous computation, we hope to generally expand the toolkit of techniques that can be applied to analyzing graph algorithms in asynchronous systems. 

As a motivating application, we  use our results to give state-of-the-art distributed community detection protocols, significantly improving upon prior work for the well-studied  stochastic-block model and related models where nodes communicate more frequently within their community than outside of it. We summarize our results below.
%

%

\medskip
\noindent\textbf{Asynchronous eigenvector computation.} First, we provide an algorithm (\ripref{alg:ojaGossip}) that 
approximates the $k$ largest eigenvectors $\bv{v}_1,...,\bv{v}_k$ for an arbitrary communication matrix (essentially a normalized adjacency matrix, defined formally in \autoref{def:weightedPopProtocol}). 

For an $n$-node network,
the algorithm ensures, with good probability, that each node $u \in [n]$ computes the $u^{th}$ entries of vectors $\bv{ \tilde v}_1,...,\bv{\tilde v}_k$ such that for all $i \in [k]$, $\norm{\bv{\tilde v}_i-\bv{v}_i}_2^2 \leq \epsilon$.
Each message sent by the algorithm requires communicating just $O(k)$ numbers, and the global time complexity is $\tilde O( \frac{\Lambda k^3}{\gap \cdot \min(\gap, \gamma_{mix}) \epsilon^3}  )$ local rounds, where $\gap$ is  the minimal gap between the $k$ largest eigenvalues, 
$\gamma_{mix}$ is roughly speaking the spectral gap, i.e., the difference between the largest and second-largest eigenvalue, and $\Lambda$ is the sum of the $k$ largest eigenvalues. We note that we use $\tilde O(\cdot )$ to suppress  logarithmic terms, and in particular, factors of $\poly \log n$.
See \autoref{cor:ojaAsync2} for a more precise statement. 

For illustration, consider a communication graph generated via the stochastic block model -- $G(n,p,q)$, which has $n$ nodes, partitioned into two equal-sized clusters. Each intracluster edge added independently with probability $p$ and each intercluster edge is added with probability  $q < p$. If, for example, $p = \Omega \left(\frac{\log n}{n}\right)$ and $q = p/2$, and $k = 2$, we can bound with high probability $\Lambda = \Theta(1/n)$, $\gap = \Theta(1/n)$, and $\gamma_{mix} = \Theta(1/n)$, which yields an eigenvector approximation algorithm running in $\tilde O( \frac{n}{\epsilon^3}  )$ global rounds, or 
  $\tilde O( \frac{1}{\epsilon^3}  )$ local rounds.

\medskip
\noindent\textbf{Approximate community detection.}
Second, we harness our eigenvector approximation routine for community detection in the stochastic block model with connection probabilities $p,q$ (we give two natural definitions of this model in an asynchronous distributed system with a random scheduler; see Definitions \ref{def:bpp} and \ref{def:sbm}).
After executing our protocol (\ripref{alg:commDetec}), with good probability, all but an $\epsilon$ fraction of the nodes output a correct community label in
$\tilde O \left (1/\epsilon^3  \rho^2 \right)$ local rounds, where $\rho = \min \left (\frac{q}{p+q}, \frac{p-q}{p+q} \right )$.
For example, when $q= p/2$, this complexity is $\tilde O \left (1/\epsilon^3 \right)$.
 See \autoref{thm:modelp} and \autoref{thm:modelg} for precise bounds.
 
 \medskip
\noindent\textbf{Exact community detection.}
 Finally, we show how to produce an exact community labeling, via a simple gossip-based error correction scheme.
 For ease of presentation, here we just state  our results in the case when  $q= p/2$ and we refer to 
  \autoref{sec:cleanup} (Theorems \ref{thm:cleanuppq} and \ref{thm:cleanupg}) for general results. Starting from an approximate labeling in which only a small constant fraction of the nodes are incorrectly labeled, we show that, with high probability, after  $O(\log n)$ local rounds, all nodes are labeled correctly.


\vspace{-2ex}
\paragraph{Related work}
\vspace{-1ex}
Community detection via graph eigenvector computation and other spectral methods has received ample attention in  centralized setting \cite{McS01,ColeFR15,vu14simple}. Such methods are known to recover communities in the stochastic block model close to the information theoretic limit. Interestingly, many state-of-the-art community detection algorithms in this model, which improve upon spectral techniques, are based on message passing (belief propagation) algorithms \cite{decelle2011asymptotic,mossel2014belief}. However, these algorithms are not known to work in asynchronous contexts.


Community detection in asynchronous distributed systems has received less attention. It has recently been tackled in a beautiful paper by Becchetti et al. \cite{becchetti2017friend}. The algorithm studied in this paper is a very simple averaging protocol, originally considered by the authors in a synchronous setting \cite{becchetti2017find}. 
Each node starts with a random value chosen uniformly in $\{-1,1\}$.
Each time two nodes communicate, they update their values to the average of their previous values. After each round of communication, a node's estimated community is given by the \emph{sign} of the change of its value due to the averaging update in that round. 

Beccheti et al. analyze their algorithm for \emph{regular} clustered graphs, including regular stochastic block model graphs, where all nodes have exactly $a$ edges to (randomly selected) nodes in their cluster and exactly  $b < a$ edges to nodes outside their cluster. As discussed in \cite{becchetti2017friend}, for regular graphs their protocol can be viewed as estimating the sign of entries in the second largest adjacency matrix eigenvector. Thus, it has close connections with our protocols, which explicitly  estimate  this eigenvector and label communitues using the signs of its entries.


The results of Becchetti et al. apply with $O(\plog n)$ local rounds of communication when either $\frac{a}{b} = \Omega(\log^2 n)$, or when $a-b = \Omega(\sqrt{a+b})$. In contrast, our results for the (non-regular) stochastic block model give $O(\plog n)$ local runtime when $\frac{p}{q} = \Omega(1)$ or $n(p-q) = \Omega( \sqrt{n(p+q) \log n})$. Here we assume that $q$ is not too small -- see \autoref{thm:modelg} for details. Note that $n\cdot p$ and $n\cdot q$ can be compared to $a$ and $b$, since they  are the expected number of intra- and inter-cluster edges respectively. Thus, our results give comparable bounds, tightening those of Becchetti et al. in some regimes and holding in the most commonly studied family of stochastic block model graphs, without any assumption of regularity\footnote{We note that the analysis of Bechitti et al. seems likely to extend to our alternative communication model (\autoref{def:bpp}), where the communication graph is weighted and regular}.




Outside of community detection, our approach to asynchronous eigenvector approximation is related to work on asynchronous distributed stochastic optimization \cite{tsitsiklis1986distributed,de2015taming,recht2011hogwild}. Often, it is assumed that many processors update some decision variable in parallel. If these updates are sufficiently sparse, overwrites are rare and the algorithm converges as if it were run in a synchronous manner. Our implementation of Oja's algorithm falls under this paradigm. Each update to our eigenvector estimates is sparse -- requiring a modification just by the two nodes that communicate at a given time. In this way, we can fully parallelize the algorithm, even in an asynchronous system.

\section{Preliminaries}

\subsection{Notation}
\label{sec:notation}
For integer $n>0$, let $[n] \eqdef \{1,\ldots, n\}$. Let $\bv{1}_{n,m}$ be an $n \times m$ all-ones matrix  and $\bv{I}_{n\times n}$ be an $n\times n$ identity. Let $\bv{e}_i$ be the $i^\text{th}$ standard basis vector, with length apparent from context. 
Let $V$ denote a set of nodes with cardinality $|V| = n$. 
Let $\pairs$ be the set of all unordered node pairs $(u,v)$ with $u \neq v$. $|\pairs| = {n\choose 2}$. 

For vector $\bv{x} \in \R^n$, $\norm{\bv{x}}_2$ is the Euclidean norm. For matrix $\bv{M} \in \R^{n \times m}$, $\norm{\bv{M}}_2 = \max_{\bv{x}} \frac{\norm{\bv{M}\bv{x}}_2}{\norm{\bv{x}}_2}$ is the spectral norm. $\norm{\bv{M}}_F = \sqrt{\sum_{i=1}^n \sum_{j=1}^m \bv{M}_{i,j}^2}$ denotes the Frobenius norm. $\bv{M}^T$ is the matrix transpose of $\bv{M}$. When $\bv{M} \in \R^{n \times n}$ is symmetric we let $\lambda_1(\bv{M}) \ge \lambda_2(\bv{M}) \ge ... \ge \lambda_n(\bv{M})$ denote its eigenvalues. $\bv{M}$ is positive semidefinite (PSD) if $\lambda_i(\bv{M}) \ge 0$ for all $i$. For symmetric $\bv{M},\bv{N}\in\R^{n\times n}$ we use $\bv{M} \preceq \bv{N}$ to indicate that $\bv{N}-\bv{M}$ is PSD.
\vspace{-2ex}
\subsection{Computational model}
\vspace{-1ex}
\label{sec:comp_model}
We define an asynchronous distributed computation model that encompasses both the well-studied population protocol \cite{aspnes2009introduction} and asynchronous gossip models \cite{boyd2006randomized}. Computation proceeds in rounds and a random scheduler chooses a single pair of nodes to communicate in each round. The choice is independent across rounds, but may be nonuniform across node pairs.

\begin{definition}[Asynchronous communication model]\label{def:weightedPopProtocol}
Let $V$ be a set of nodes with $|V| = n$. Computation proceeds in rounds, with every node $v \in V$ having some state $s(v,t)$ in round $t$.

Recall that $\pairs$ denotes all unordered pairs of nodes in $V$. Let $w: \pairs \rightarrow \mathbb{R}^+$ be a nonnegative weight function. 
In each round, a random scheduler chooses exactly one $(u,v) \in \pairs$ with probability $w(u,v)/\left[\sum_{(i,j) \in \pairs} w(i,j)\right]$ and $u,v$ both update their states according to some common (possibly randomized) transition function $\sigma$. Specifically, they set $s(v,t+1) = \sigma(s(v,t),s(u,t))$ and $s(u,t+1) = \sigma(s(u,t),s(v,t))$. 
\end{definition}

Note that in our analysis we often identify the weight function $w$ with a symmetric weight matrix $\bv{W} \in \mathbb{R}^{n \times n}$ where $\bv{W}_{u,u} = 0$ and $\bv{W}_{u,v} = \bv{W}_{v,u} = w(u,v)/\left[\sum_{(i,j) \in \pairs} w(i,j)\right]$. Let $\bv{D}$ be a diagonal matrix with $\bv{D}_{u,u} = \sum_{v \in V} \bv{W}_{u,v}$. $\bv{D}_{u,u}$ is the probability that node $u$ communicates in any given round. Since two nodes are chosen in each round,  $\sum_u \bv{D}_{u,u} = 2$. We will refer to $\bv{D} + \bv{W}$ as the \emph{communication matrix} of the communication model.

%

\begin{remark}[Asynchronous algorithms]
Since the transition function $\sigma$ in \autoref{def:weightedPopProtocol} is universal, nodes can be seen as identical processes, with no knowledge of $w$ or unique ids. We do assume that nodes can initiate and terminate a protocol synchronously. That is, nodes interact from round $0$ up to some round $T$, after which they cease  to interact, or begin a new protocol. This assumption is satisfied if each node has knowledge of the global round number but, in general, is much weaker. For example, in the asynchronous gossip model discussed below, it is sufficient for nodes to have access to a synchronized clock.


We use \emph{algorithm} to refer to a sequence of transition functions, each corresponding to a subroutine run for specified number of rounds.  Subroutines are run sequentially. The first has input nodes with identical starting states (as prescribed by \autoref{def:weightedPopProtocol}) but later subroutines start once nodes have updated their states and thus have distinguished inputs.

\end{remark}

\begin{remark}[Simulation of existing models]
The standard \emph{population protocol model} \cite{aspnes2009introduction} is recovered from \autoref{def:weightedPopProtocol} by setting $w(u,v) = 1$ for all $(u,v)$ -- i.e., pairs of nodes communicate uniformly at random. A similar model over a fixed communication graph $G = (E,V)$ is recovered by setting $w(u,v)=1$ for all $(u,v) \in E$ and $w(u,v) = 0$ for $(u,v) \notin E$.

\autoref{def:weightedPopProtocol} also encompasses the \emph{asynchronous gossip model} \cite{boyd2006randomized,dimakis2010gossip}, where each node holds an independent Poisson clock and contacts a random neighbor when the clock ticks. If we identify rounds with clock ticks, let $\lambda_u$ be the rate of node $u$'s clock, and let $p(u,v)$ be the probability that $u$ contacts $v$ when its clock ticks. Then the probability that nodes $u$ and $v$ interact in a given round is $\frac{1}{2}\left[\frac{\lambda_u}{\sum_{z \in V} \lambda_z} \cdot p(u,v) + \frac{\lambda_v}{\sum_{z \in V} \lambda_z} \cdot p(v,u)\right]$. With $w(u,v)$ set to this value, \autoref{def:weightedPopProtocol} corresponds exactly to the asynchronous gossip model.
\end{remark}

\subsection{Distributed community detection problem}

This paper studies the very general problem of computing communication matrix eigenvectors with asynchronous protocols run by the nodes in $V$. One primary application of computing eigenvectors is to detect community structure in $G$. Below we formalize this application as the \emph{distributed community detection problem} and introduce two specific cases of interest.

In the distributed community detection problem, the weight function $w$ and corresponding weight matrix $\bv{W}$ of \autoref{def:weightedPopProtocol} are clustered: nodes in the same cluster are more likely to communicate than nodes in different clusters. The goal is for each node to independently identify what cluster it belongs to (up to a permutation of the cluster labels). 

We consider two models of clustering. In the first \emph{\modelpq}, the weight function directly reflects the increased likelihood of intracluster communication. In the second, \emph{\modelg}, weights are uniform on a graph sampled from the well-studied planted-partition or \emph{stochastic block model} \cite{holland1983stochastic}.
For simplicity, we focus on the setting in which there are two equal sized clusters, but believe that our techniques can be extended to handle a larger number of clusters, potentially with unbalanced sizes.

\begin{definition}[\modelpq]
\label{def:bpp}
An asynchronous model (\autoref{def:weightedPopProtocol}), where node set $V$ is partitioned into disjoint sets $V_1,V_2$ with $|V_1| = |V_2| = n/2$. For values $q < p$, $w(u,v) = p$ if $u,v \in V_i$ for some $i$ and $w(u,v) = q$ if $u \in V_i$ and $v \in V_j$ for $i \neq j$.
\end{definition}

%
%
\begin{definition}[\modelg]\label{def:sbm} An asynchronous model (\autoref{def:weightedPopProtocol}), where node set $V$ is partitioned into disjoint sets $V_1,V_2$ with $|V_1| = |V_2| = n/2$. The weight matrix $\bv{W}$ is a normalized adjacency matrix of a random graph $G(V,E)$ generated as follows:
 for each pair of nodes $u,v \in V$, add edge $(u,v)$ to edge set $E$ with probability $p$ if $u$ and $v$ are in the same partition $V_i$ and probability $q < p$ if $u$ and $v$ are in different partitions. 
\end{definition}
Analysis of community detection in the \modelpq\ is more elegant, and will form the basis of our analysis for the \modelg, which more closely matches models considered in prior work on in both distributed and centralized settings. Formally, 
we define the distributed community detection problem as follows:
\begin{definition}[Distributed community detection problem]\label{def:cc}
An algorithm executing in the communication models of  \autoref{def:bpp} and \autoref{def:sbm} solves community detection in $T$ rounds if for every $t \ge T$, 
all nodes in $V_1$ hold some integer state $s_1 \in \{-1,1\}$, while all nodes in $V_2$ hold state $s_2 = -s_1$.
An algorithm solves the community detection problem in $L$ \emph{local rounds} if every node's state remains fixed after $L$ local interactions with other nodes.
\end{definition}

\section{Asynchronous Oja's algorithm}\label{sec:sec3}
Our main contribution is a distributed algorithm for computing eigenvectors of the communication matrix $\bv{D} + \bv{W}$. These eigenvectors can be used to solve the distributed community detection problem or in other applications. Our main algorithm is a distributed, asynchronous adaptation of Oja's classic iterative eigenvector algorithm \cite{Oja1982}, described below:
\begin{algorithm}[H]
\caption{\algoname{Oja's method (centralized)}}
{\bf Input}: $\bv{x}_0,...,\bv{x}_{T-1} \in \mathbb{R}^{n}$ drawn i.i.d. from some distribution $\dist$ such that for some constant $C$, $\Pr_{\bv{x} \sim \dist}[\norm{\bv{x}}_2^2 \le C] = 1$ and $\E_{\bv{x} \sim \dist}[\bv{x} \bv{x}^T] = \bv{M}$. Rank parameter $k$ and step size $\eta$.\\
{\bf Output}: Orthonormal $\bv{\tilde V} \in \mathbb{R}^{n \times k}$ whose columns approximate $\bv{M}$'s $k$ top eigenvectors.
\begin{algorithmic}[1]
\State{Choose $\bv{Q}_0$ with entries drawn i.i.d. from the standard normal distribution $\mathcal{N}(0,1)$.}
\For{$t = 0,....,T-1$}
\State{$\bv{Q}_{t+1} := (\bv{I} + \eta \bv{x}_t\bv{x}_t^T)\bv{Q}_t$.}
\EndFor\\
\Return{$\bv{\tilde V}_T := \orth(\bv{Q}_T)$.}
\Comment{\textcolor{blue}{Orthonormalizes the columns of $\bv{Q}_T$.}}
\end{algorithmic}
\label{alg:oja}
\end{algorithm}

\ifArxiv
Note that in \ripref{alg:oja}, the step size $\eta$ remains fixed throughout the algorithm. It is possible to achieve improved runtime bounds if $\eta$ is set differently in different rounds (see \cite{allen2016first}).
However, for the sake of simplicity, in a distributed setting in which updates are performed asynchronously, we consider a fixed step size. This only affects our runtime bounds for community detection by constant factors. 
We note that the orthonormalization step (Step 5) is often performed at each iteration, after $\bv{Q}_{t+1}$ is computed. However, as shown in \cite{allen2016first}, orthogonalizing $\bv{Q}_T$ at the end of the iterations suffices.
\fi 
\vspace{-2ex}
\subsection{Approximation bounds for Oja's method}
A number of recent papers have provided strong convergence bounds for the \emph{centralized} version of Oja's method \cite{allen2016first,aaronsPaper}. We will rely on the following theorem, which we prove in \autoref{app:oja} using a straightforward application of the arguments in \cite{allen2016first}.
\begin{theorem}\label{thm:oja}
Let $\bv{M} \in\mathbb{R}^{n \times n}$ be a PSD matrix with $\frac{\sum_{i=1}^k \lambda_i(\bv{M})}{C} \le \Lambda$ and $\frac{\lambda_k(\bv{M}) - \lambda_{k+1}(\bv{M})}{C} \ge \gap$ for some values $\Lambda,\gap$. For any $\epsilon,\delta \in (0,1)$, let $\xi = \frac{n}{\delta \epsilon \cdot \gap}$, $\eta = \frac{c_1\epsilon^2 \cdot \gap \cdot \delta^2}{C\Lambda k \log^3 \xi}$ for some sufficiently small constant $c_1,$ and $T = \frac{c_2\cdot(\log \xi+1/\epsilon)}{C\cdot\gap \cdot \eta}$ 
for sufficiently large $c_2$. 
\ifArxiv{
Then with probability $\ge 1 -\delta$, \ripref{alg:oja} run with step size $\eta$ satisfies,  
for all $t \in [T-\frac{c_2/\epsilon}{2C\gap \cdot \eta},T+\frac{c_2/\epsilon}{2C\gap \cdot \eta}]$, letting $\bv{\tilde V}_t = \orth(\bv{Q}_t)$:
\else
Then with probability $\ge 1 -\delta$, \ripref{alg:oja} run with step size $\eta$ returns $\bv{\tilde V}_T$ satisfying, 
\fi
 $$\norm{\bv{Z}^T \bv{\tilde V}_T}_F^2 \le \epsilon.$$
 where $\bv{Z}$ is an orthonormal basis for the bottom $n-k$ eigenvectors of $\bv{M}$.
\end{theorem}
If $\bv{\tilde V}_T$ exactly spanned $\bv{M}$'s top $k$ eigenvectors, $\norm{\bv{Z}^T \bv{\tilde V}_T}_F^2$ would equal 0. To obtain an approximation of $\epsilon$, the number of iterations required by Oja's method naturally depends inversely on $\epsilon$, the failure probability $\delta$, and the gap between eigenvalues $\lambda_k(\bv{M})$ and $\lambda_{k+1}(\bv{M})$.

\ifArxiv
In particular, setting $t = T$ in \autoref{thm:oja} implies that  the output $\bv{\tilde V}$ of \ripref{alg:oja} satisfies $\norm{\bv{Z}^T \bv{\tilde V}}_F^2 \le \epsilon.$ However, we state the more general guarantee above since when implemented in some asynchronous settings, the number of rounds before termination may vary about the target $T$. For example, nodes may  all terminate at  a specific time, once $T$ communication rounds have occurred \emph{in expectation}, but not after a precise number of rounds.
\fi

\subsection{Distributed Oja's method via random edge sampling}
\label{dist_ojas}

Oja's method can be implemented in the asynchronous communication model (\autoref{def:weightedPopProtocol}) to compute top eigenvectors of the communication matrix $\bv{D}+\bv{W}$, defined in \autoref{sec:comp_model}.

For any pair of nodes $(u,v)$, let $\euv = \bv{e}_u + \bv{e}_v$ be the vector with all zero entries except $1$'s in its $u^\text{th}$ and $v^\text{th}$ positions. Given weight function $w$ and associated matrix $\bv{W}$, let $\dw$ be the distribution in which each $\euv$ is selected with probability $\bv{W}_{u,v}$. That is, the same distribution by which edges are selected to be active by the scheduler in \autoref{def:weightedPopProtocol}. Noting that $\euv \euv^T$ is all zero except at its $(u,u)$, $(v,v)$, $(u,v)$, and $(v,u)$ entries, we can see that
\begin{align}
\E_{\euv \sim \dw} \left [\euv \euv^T \right ] &= \sum_{(u,v)\in \mathcal{P}} \bv{W}_{u,v} \cdot \euv \euv^T
= \bv{D} + \bv{W}\label{eq:basicExp},
\end{align}
 where $\pairs$ denotes the set of unordered node pairs $(u,v)$ with $u\neq v$.
So if we run Oja's algorithm with $\euv$ sampled according to $\dw$, we will obtain an approximation to the top eigenvectors of $\bv{D} + \bv{W}$. Note that this matrix is PSD, by the fact that each $\euv \euv^T$ is PSD.

Furthermore, the algorithm can be implemented in our communication model as an \emph{extremely simple averaging protocol.}
Each iteration of \ripref{alg:oja} requires computing $\bv{Q}_{t+1} = (\bv{I} + \eta \bv{x}_t\bv{x}_t^T) \bv{Q}_t$. If $\bv{x}_t = \euv$ for $\euv \sim \dw$, we can see that computing $\bv{Q}_{t+1}$ just requires updating the $u^\text{th}$ and $v^\text{th}$ rows of $\bv{Q}_t$. Thus, if the $n$ rows of $\bv{Q}_t$ are distributed across $n$ nodes, this update can be done locally by nodes $u$ and $v$ when they are chosen to interact by the randomized scheduler. Specifically, letting $[q^{(1)}_u,...,q^{(k)}_u]$ be the $u^\text{th}$ row of $\bv{Q}_t$, stored as the state at node $u$, applying $(\bv{I} + \eta \euv \euv^T)$ just requires setting for all $i \in [k]$:
\begin{align}\label{eq:avg}
q^{(i)}_u := (1+\eta)q^{(i)}_u + \eta q^{(i)}_v.
\end{align}
Node $v$ makes a symmetric update, and all other entries of $\bv{Q}_t$ remain fixed.

We give the pseudocode for this protocol in \ripref{alg:ojaGossip}. Along with the main iteration based on the simple update in \eqref{eq:avg}, the nodes 
need to implement Step 5 of \ripref{alg:oja}, where $\bv{Q}_T$ is orthogonalized. This can be done with a gossip-based protocol, which we abstract as the routine $\mathtt{AsynchOrth}$. We give an implementation of  $\mathtt{AsynchOrth}$ in \autoref{sec:orth}. 

\begin{remark}[Choice of communication matrix]
 While, as we will show, the eigenvectors of $\bv{D} + \bv{W}$ are naturally useful in our applications to community detection, the above techniques easily extend to computing eigenvectors of other matrices. For example, if we set $\euv = \bv{e}_u - \bv{e}_v$,  $\E_{\euv \sim \dw}[\euv \euv^T] = \bv{D}-\bv{W} = \bv{L}$, a scaled Laplacian of the communication graph. 
 \ifArxiv
 Alternatively, if the nodes perform an initialization protocol in which they compute the maximum node degree $\Delta = \max_u \bv{D}_{u,u}$ via gossiping, they can compute the eigenvectors of the PSD matrix $\bv{\Delta}\bv{I} + \bv{W}$ which are identical to those of $\bv{W}$. When $\bv{W}$ represents the adjacency matrix of some communication graph, these are the adjacency matrix eigenvectors.
 \fi
 \end{remark}

\begin{algorithm}[H]
\caption{\algoname{Asynchronous Oja's} (\texttt{AsynchOja}$(T,T',\eta)$)}
{\bf Input}: Time bounds $T,T'$, step size $\eta$.\\
{\bf Initialization}: $\forall u$, chose $[q^{(1)}_u,...,q^{(k)}_u]$ independently from standard Gaussian $\mathcal{N}(0,1)$.
\begin{algorithmic}[1]
\If{$t < T$}
\State{$(u,v)$ is chosen by the randomized scheduler.}
\State{For all $i \in [k]$, $q^{(i)}_u := (1+\eta)q^{(i)}_u + \eta q^{(i)}_v$.}
\Comment{\textcolor{blue}{Computes of $(\bv{I} + \eta \euv \euv^T)\bv{Q}_t$}.}
\Else
\State{$[{\hat v}_u^{(1)},...,{\hat v}_u^{(k)}] = \texttt{AsynchOrth}([{q}_u^{(1)},...,{q}_u^{(k)}],T')$.}
\Comment{\textcolor{blue}{Implements of $\bv{\tilde V}_T = \orth(\bv{Q}_T)$.}}
\EndIf
\end{algorithmic}
\label{alg:ojaGossip}
\end{algorithm}
Note that in the pseudocode above, when nodes $u,v$ interact in the asynchronous model, they only need to share their respective values of $q_u^{(i)}$ and $q_v^{(i)}$ for $i\in[k]$.

Up to the orthogonalization step, we see that \ripref{alg:ojaGossip} \emph{exactly simulates} \ripref{alg:oja} on input $\bv{M} = \bv{D} +  \bv{W}$. Thus, assuming that 
%
$\texttt{AsynchOrth}([{q}_u^{(1)},...,{q}_u^{(k)}])$ exactly computes $\bv{\tilde V}_T = \orth(\bv{Q}_T)$ as in Step 5 of \ripref{alg:oja}, the error bound of \autoref{thm:oja} applies directly. Specifically, if we let the local states, $[q_1^{(1)}, \ldots, q_n^{(1)}], \ldots, [q_1^{(k)}, \ldots, q_n^{(k)}]$ correspond to the $k$ length-$n$ vectors in $\bv{\tilde V}_T$, \autoref{thm:oja} shows that $\norm{\bv{Z}^T\bv{\tilde V}_T}_F^2 \le \epsilon$. In \autoref{sec:orth} we show that this bound still holds when $\texttt{AsynchOrth}$ computes an approximate orthogonalization.

\subsection{Distributed orthogonalization and eigenvector guarantees}\label{sec:orth}
%

In fact, a specific orthogonalization strategy yields a stronger bound, which is desirable in many applications, including community detection:   \ripref{alg:ojaGossip} can actually well approximate \emph{each} of $\bv{D}+\bv{W}$'s top $k$ eigenvectors, instead of just the subspace they span.


Specifically, let $\bv{\tilde v}_i$ denote the $i^\text{th}$ column of $\bv{\tilde V}_T$ and $\bv{v}_i$ denote the $i^\text{th}$ eigenvector of $\bv{D} + \bv{W}$. We want $(\bv{\tilde v}_i^T \bv{v}_i)^2 \ge 1- \epsilon$ for all $i$. Such a guarantee requires sufficiently large gaps between the top $k$ eigenvalues, so that their corresponding eigenvectors are identifiable. If these gaps exist, the guarantee can by using the following orthogonalization procedure:

\begin{algorithm}[H]
\caption{\algoname{Orthogonalization via Cholesky Factorization (centralized)}}
{\bf Input}: $\bv{Q} \in \R^{n \times k}$ with full column rank.
{\bf Output}: Orthonormal span for $\bv{Q}$, $\bv{\tilde V} \in \mathbb{R}^{n \times k}$.
\begin{algorithmic}[1]
\State{$\bv{L} := \chol(\bv{Q}^T \bv{Q})$}
\Comment{\textcolor{blue}{Cholesky decomp. returns lower triangular $\bv{L}$ with $\bv{LL}^T = \bv{Q}^T\bv{Q}$.}}\\
\Return{$\bv{\tilde V} := \bv{Q} (\bv{L}^T)^{-1}$}
\Comment{\textcolor{blue}{Orthonormalize $\bv{Q}_{T}$'s columns using the Cholesky factor.}}
\end{algorithmic}
\label{alg:chol}
\end{algorithm}
\begin{remark}
\ripref{alg:chol} requires an input that is \emph{full-rank}, which always includes $\bv{Q}_T$ in Algorithms \ref{alg:oja} and 
\ref{alg:ojaGossip}: $\bv{Q}_0$'s entries are random Gaussians so it is full-rank with probability $1$ and each $(\bv{I} + \eta \bv{x}_t^T\bv{x}_t)$ is full-rank since $\eta < \norm{\bv{x}_t}$. Thus, $\bv{Q}_T = \prod_{t=0}^{T-1} (\bv{I} + \eta \bv{x}_t^T\bv{x}_t) \bv{Q}_0$ is too.
\end{remark}

Ultimately, our $\texttt{AsynchOrth}$ is an asynchronous distributed implementation of \ripref{alg:chol}. We first prove an eigenvector approximation bound under the assumption that this implementation is exact (\autoref{cor:ojaAsync} in \autoref{app:orth}) and then adapt that result to account for the fact that $\texttt{AsynchOrth}$ only outputs an approximate solution.

Pseudocode for \texttt{AsynchOrth} is included below.
Each node first computes a (scaled) approximation to every entry of $\bv{ Q}^T \bv{ Q}$ using a simple averaging technique. Nodes then locally compute $\bv{L} = \chol\left (\bv{ Q}^T \bv{ Q}\right)$ and the $u^\text{th}$ row of $\bv{\tilde V}_T = \bv{Q} (\bv{L}^T)^{-1}$. In \autoref{app:orth} we argue that, due to numerical stability of Cholesky decomposition, each node's output is close to the $u^\text{th}$ row of an exactly computed $\bv{\tilde V}_T$, despite the error in constructing $\bv{ Q}^T \bv{ Q}$.

\begin{algorithm}[H]
\caption{\algoname{Asynchronous Cholesky Orthogonalization} (\texttt{AsynchOrth}$(T)$)}
{\bf Input}: Time bound $T$.\\
{\bf Initialization}: Each node holds $[{q}^{(1)}_u,...,{q}^{(k)}_u]$. For all $i,j \in [k]$, let $r_u^{(i,j)} := {q}^{(i)}_u \cdot {q}^{(j)}_u$.
\begin{algorithmic}[1]
\If{$t < T$}
\State{$(u,v)$ is chosen by the randomized scheduler.}
\State{for all $i,j \in [k]$, $r_u^{(i,j)} := \frac{r_u^{(i,j)}+r_v^{(i,j)}}{2}$.}
\Comment{\textcolor{blue}{Estimation of $\frac{1}{n}\bv{ q}_i^T \bv{ q}_j$ via averaging.}}
\Else
\State{Form $\bv{R}_u \in \R^{k \times k}$ with $(\bv{R}_u)_{i,j} = (\bv{R}_u)_{j,i} := n \cdot r_u^{(i,j)}$.}
\Comment{\textcolor{blue}{Approximation of $\bv{ Q}^T\bv{ Q}$.}}
\State{$\bv{L}_u := \chol(\bv{R}_u)$.}
\State $[\hat v_u^{(1)},...,\hat v_u^{(k)}] := [{q}^{(1)}_u,...,{q}^{(k)}_u] \cdot (\bv{L}_u^T)^{-1}.$
\Comment{\textcolor{blue}{Approximation of $u^\text{th}$ row of $\bv{ Q}(\bv{L}_u^T)^{-1}$.}}
\EndIf
\end{algorithmic}
\label{alg:pporth}
\end{algorithm}

In \autoref{app:orth} we prove the following result when \ripref{alg:pporth} is used to implement $\texttt{AsynchOrth}$ as a subroutine for \ripref{alg:ojaGossip}, \texttt{AsynchOja}$(T,T',\eta)$:

\begin{shaded}
\vspace{-.75em}
\begin{theorem}[Asynchronous eigenvector approximation]\label{cor:ojaAsync2}
Let $\bv{v}_1,...,\bv{v}_k$ be the top $k$ eigenvectors of the communication matrix $\bv{D}+\bv{W}$ in an asynchronous communication model, and let $\Lambda,\bgap, \gamma_{mix}$ be  bounds satisfying: $\Lambda \ge \sum_{j=1}^k \lambda_j(\bv{D}+\bv{W})$, $\bgap \le \min_{j \in [k]} [\lambda_j(\bv{D}+\bv{W})-\lambda_{j+1}(\bv{D}+\bv{W})]$, and $\gamma_{mix} \le \min \left [ \frac{1}{n},\log \left(\lambda_2^{-1}(\bv{I}-\frac{1}{2}\bv{D} + \frac{1}{2}\bv{W})\right)\right ]$.

For any $\epsilon,\delta \in(0,1)$, let $\xi = \frac{n}{\delta \epsilon \cdot \bgap}$. Let $\eta =\frac{c_1 \epsilon^2 \cdot \bgap \cdot \delta^2}{\Lambda k^3 \log^3 \xi}$ for sufficiently small $c_1$, and $T = \frac{c_2 \cdot (\log \xi + 1/\epsilon)}{\bgap \cdot \eta}$, $T' = \frac{c_3(\log \xi + 1/\epsilon)\cdot \lambda_1(\bv{D}+\bv{W})}{\bgap \cdot \gamma_{mix}}$ for sufficiently large $c_2,c_3$.
 For all $u \in [n],i \in [k]$, let ${\hat v}_u^{(j)}$ be the local state computed by \ripref{alg:ojaGossip}.
If $\bv{\hat V} \in \R^{n \times k}$ is given by $(\bv{\hat V})_{u,j} = {\hat v}_u^{(j)}$
 and $\bv{\hat v}_i$ is the $i^{th}$ column of $\bv{\hat V}$, then with probability $\ge 1 -\delta-e^{-\Theta(n)}$, for all $i \in [k]$:
 \begin{align*}
 \left |\bv{\hat v}_i^T \bv{v}_{i}\right| \ge 1-\epsilon\hspace{1em}\text{ and }\hspace{1em}\norm{\bv{\hat v}_i}_2 \le 1+\epsilon.
 \end{align*}
\end{theorem}
\vspace{-.75em}
\end{shaded}

\section{Distributed community detection}
From the results of \autoref{sec:sec3}, we obtain a simple population protocol for distributed community detection that works for many clustered communication models, including the \modelpqshort\ and \modelgshort\ models of Definitions \ref{def:bpp} and \ref{def:sbm}.

In particular, we show that if each node $u\in V$ can locally compute the $u^\text{th}$ entry of an approximation $\bv{\hat v}_2$ to the second eigenvector of the communication matrix $\bv{D} + \bv{W}$, then it can solve the community detection problem locally: $u$ just sets its state to the sign of this entry. 

\begin{algorithm}[H]
\caption{\algoname{Asynchronous Community Detection} (\texttt{AsynchCD}$(T,T',\eta)$)}
{\bf Input}: Time bounds $T,T'$, step size $\eta$.
\begin{algorithmic}[1]
\State{Run \texttt{AsynchOja}$(T,T',\eta)$ (\ripref{alg:ojaGossip}) with $k = 2$.}
\State{Set $\hat \chi_u := \sign(\hat v_u^{(2)}).$}
\end{algorithmic}
\label{alg:commDetec}
\end{algorithm}
Here $\hat \chi_u \in \{-1,1\}$ is the final state of node $u$. We will claim that this state solves the community detection problem of \autoref{def:cc}. We use the notation $\hat \chi_u$ because we will use $\bs{\chi}$ to denote the true \emph{cluster indicator vector} for communities $V_1$ and $V_2$ in a given communication model: $\bs{\chi}_u =1$ for $u \in V_1$ and $\bs{\chi}_u = -1$ for $u \in V_2$.

In particular, we will show that if $\eta$ is set so that \texttt{AsynchOja} outputs eigenvectors with accuracy $\epsilon$, then a $1-O(\epsilon)$ fraction of nodes will correctly identify their clusters. 
In \autoref{sec:cleanup} we show how to implement a 	`cleanup phase' where, starting with $\epsilon$ set to a small constant (e.g. $\epsilon = .1$), the nodes can converge  to a state with \emph{all} cluster labels correct with high probability. 

\subsection{Community detection in the \modelpq}

We start with an analysis for the \modelpq.
Recall that in this model the nodes are partitioned into two sets, $V_1$ and $V_2$, each with $n/2$ elements. Without loss of generality we can identify  the nodes with integer labels such that $1, \ldots, n/2 \in V_1$ and $n/2 + 1, \ldots, n \in V_2$. 
We define the weighted cluster indicator matrix, $\bv{C}^{(p,q)} \in \R^{n\times n}$:
\begin{align}\label{eq:indicatorMatrix}
\bv{C}^{(p,q)} \eqdef \begin{bmatrix} 
p \cdot \bv{1}_{\frac{n}{2} \times \frac{n}{2}} & q \cdot \bv{1}_{\frac{n}{2} \times \frac{n}{2}} \\ 
q \cdot \bv{1}_{\frac{n}{2} \times \frac{n}{2}} & p \cdot \bv{1}_{\frac{n}{2} \times \frac{n}{2}}
\end{bmatrix}.
\end{align}
$p$ and $q$ can be arbitrary, but we will always take $p > q > 0$. 
It is easy  to check that $\bv{C}^{(p,q)}$ is a rank two matrix with eigendecomposition:
\begin{align}
\label{spect_c}
\bv{C}^{(p,q)} &= \frac{n}{2}\begin{bmatrix} \\ \bv{v}_1 & \bv{v}_2 \\ \\ \end{bmatrix} \begin{bmatrix} p+q & 0 \\ 0 & p-q \end{bmatrix} \begin{bmatrix} &\bv{v}_1^T & \\ & \bv{v}_2^T & \end{bmatrix} & &\text{where} & \bv{v}_1 & = \frac{\bv{1}_{n \times 1}}{\sqrt{n}},\hspace{.5em}  \bv{v}_2 = \frac{\bs{\chi}}{\sqrt{n}}.
\end{align}
So, if all nodes could compute their corresponding entry in the \emph{second eigenvector} of $\bv{C}^{(p,q)}$, then by simply returning the sign of this entry, they would solve the distributed community detection problem (\autoref{def:cc}).
If they compute this eigenvector approximately, then we can still show that a large fraction of them correctly solve community  detection. Specifically:
%
\begin{lemma}
\label{approx_eig_starting_point}
Let $\bv{v}_2$ be the second eigenvector of $\bv{C}^{(p,q)}$ for any $p > q > 0$. If $\bv{\tilde v}_2$ satisfies:
\begin{align}
\label{what_we_have}
 \left |\bv{\tilde v}_2^T \bv{v}_{2}\right| \ge 1-\epsilon\hspace{1em}\text{ and }\hspace{1em}\norm{\bv{\tilde v}_2}_2 \le 1+\epsilon.
\end{align}
for $\epsilon \le 1$,
then $\sign(\bv{\tilde v}_2)$ gives a labeling such that, after ignoring at most $5\epsilon n$ nodes, all remaining nodes in $V_1$ have the same labeling, and all in $V_2$ have the opposite. 
\end{lemma}
\begin{proof}
We follow the argument from \cite{spielmanNotes}. 
Let $\norm{\bv{y}}_0$ denote the number of non-zeros in a vector $\bv{y}$. Since $p > q > 0$, $\bv{v}_2 =\frac{\bs{\chi}}{\sqrt{n}}$ by \eqref{spect_c}, so the number of nodes misclassified by $\sign(\bv{\tilde v}_2)$ is:
$$
\min_{s\in \{-1,1\}} \|\sign(s\cdot \bv{\tilde v}_2)- \sign(\bv{v}_2) \|_0.
$$
For $s \in \{-1,1\}$, if
$\sign(s \cdot \bv{\tilde v}_2)$ and $\sign(\bv{v}_2)$ differ on a specific coordinate, then since each entry of $\bv{v}_2$ has value $\pm \frac{1}{\sqrt{n}}$, $s\cdot \bv{\tilde v}_2$ and $\bv{v}_2$ must differ by at least $\frac{1}{\sqrt{n}}$ on that coordinate. It follows that
$
\|\bv{\tilde v}_2 - s\cdot \bv{v}_2 \|_2^2 \geq \frac{1}{n}\|\sign(s \cdot \bv{\tilde v}_2) - \sign(\bv{v}_2) \|_0
$
and hence;
\begin{align}\label{l0l2Bound}
\min_{s \in \{-1,1\}}\|\sign(s \cdot \bv{\tilde v}_2) - \sign(\bv{v}_2) \|_0 \le n\cdot \min_{s \in \{-1,1\}}\|s\cdot\bv{\tilde v}_2 - \cdot \bv{v}_2 \|_2^2.
\end{align}
We can bound the righthand side of \eqref{l0l2Bound} using \eqref{what_we_have}. Specifically, assuming $\epsilon \leq 1$:
\begin{align*}
n\cdot \min_{s \in \{-1,1\}}\|s\cdot \bv{\tilde v}_2 - \bv{v}_2 \|_2^2 &= n\cdot \min_{s \in \{-1,1\}} \left (\norm{\bv{v}_2}_2^2 + \norm{\bv{\tilde v}}_2^2  - 2 (\bv{\tilde v}_i^T \bv{v}_2) \right )\\
&\le n \cdot \left (1 + (1+\epsilon)^2 - 2(1-\epsilon) \right ) \leq 5n\epsilon
\end{align*}
Plugging back into \eqref{l0l2Bound}, we have
$
\min_{s \in \{-1,1\}}\|\sign(s\cdot\bv{\tilde v}_2) - \sign(\bv{v}_2) \|_0 \leq 5 n\epsilon$,
so $\sign(\bv{\tilde v}_2)$ only classifies at most a $5\epsilon$ fraction of nodes incorrectly, giving the lemma.
\end{proof}
With \autoref{approx_eig_starting_point} in place, we can then apply \autoref{cor:ojaAsync2} to prove the correctness of \texttt{AsynchCD} (\ripref{alg:commDetec}) for the \modelpq
\begin{shaded}
\vspace{-.75em}
\begin{theorem}[$\epsilon$-approximate community  detection: \modelpq]\label{thm:modelp}
Consider \ripref{alg:commDetec} in the \modelpq. Let $\rho = \min \left (\frac{q}{p+q}, \frac{p-q}{p+q} \right )$. For sufficiently small constant $c_1$ and sufficiently large $c_2$ and $c_3$, let
\begin{align*}
\hspace{-1em}\eta =\frac{c_1 \epsilon^2 \delta^2 \rho}{\log^3 \left (\frac{n}{\epsilon \delta \rho} \right )}, \hspace{.5em}T = \frac{c_2 n \left (\log^3 \left (\frac{n}{\epsilon \delta \rho} \right ) + \frac{\log \left (\frac{n}{\epsilon \delta \rho} \right )}{\epsilon}\right)}{\epsilon^2 \delta^2 \rho^2},\hspace{.5em} T' = \frac{c_3 n  \left (\log \left (\frac{n}{\epsilon \delta \rho} \right ) + \frac{1}{\epsilon} \right )}{\rho^2}.
\end{align*}
With probability  $1-\delta$, after ignoring $\epsilon n$ nodes, all remaining nodes in $V_1$ terminate in some state $s_1 \in \{-1,1\}$, and all  nodes in $V_2$ terminate in state $s_2 = -s_1$. Suppressing polylogarithmic factors in the parameters, the total number of global rounds and local rounds required are: $T+T' = \tilde O \left ( \frac{n}{\epsilon^3 \delta^2 \rho^2} \right )$ and $L = \tilde O \left (\frac{1}{\epsilon^3 \delta^3  \rho^2} \right)$.
\end{theorem}
\vspace{-.75em}

\end{shaded}
\begin{proof}

In the \modelpq\ the weight and degree matrices are:
\begin{align*}
\bv{W} = \frac{4}{n^2(p+q) - 2n p} \cdot  (\bv{C}^{(p,q)} - p\cdot  \bv{I}_{n\times n}) \hspace{1em}\text{ and }\hspace{1em} \bv{D} = \frac{2}{n} \cdot \bv{I}_{n \times n}.
\end{align*}
Thus, referring to the eigendecomposition of $\bv{C}^{(p,q)} $ shown in \eqref{spect_c}, the top eigenvector of $\bv{D}+\bv{W}$ is $\bv{v}_1 = \bv{1}_{n \times 1}/\sqrt{n}$ with corresponding eigenvalue:
$\lambda_1 = \frac{4}{n^2(p+q) - 2n p} \cdot \left  (\frac{n(p+q)}{2}-p\right )+ \frac{2}{n} = \frac{4}{n}.$
The second eigenvector is the scaled cluster indicator vector $\bv{v}_2 = \bs{\chi}/\sqrt{n}$ with eigenvalue 
$$\lambda_2= \frac{4}{n^2(p+q) - 2n p} \cdot \left  (\frac{n(p-q)}{2}-p\right )+ \frac{2}{n}  =\frac{4}{n} \cdot \frac{p}{p+\frac{n}{n-2}\cdot q}.$$ 
Finally, for all remaining eigenvalues of $\bv{D}+\bv{W}$, $\{\lambda_3,...,\lambda_n\}$, 
$\lambda_i = \frac{2}{n} - \frac{4p}{n^2(p+q) - 2n p}.$
We can bound the eigenvalue gaps:
\begin{align*}
\lambda_1-\lambda_2&\ge \frac{4}{n} - \frac{4}{n}\cdot  \frac{p}{p+q} = \frac{4q}{n(p+q)} & & & 
\lambda_2-\lambda_3&= \frac{2(p-q)}{n(p+q)-2p} \ge \frac{2(p-q)}{n(p+q)}
\end{align*}

Let $\rho = \min \left (\frac{q}{p+q}, \frac{p-q}{(p+q)} \right )$. We bound the mixing time of $\bv{W} + \bv{D}$ by noting that $\lambda_2(\bv{I} - 1/2\bv{D} + 1/2\bv{W}) \le 1-\frac{2q}{n(p+q)}$. Then using that $\log(1/x) \ge 1-x$ for all $x \in (0,1]$, $\log(\lambda_2^{-1}(\bv{I} - 1/2\bv{D} + 1/2\bv{W}) \ge \frac{2q}{n(p+q)} \ge \frac{2\rho}{ n}$.		
We then apply \autoref{cor:ojaAsync2} with $k = 2$, $\Lambda = \frac{4}{n} + \frac{4}{n} \frac{p}{p+\frac{n}{n-2}q} \le \frac{8}{n}$, $\bgap = \frac{4}{n} \cdot \min \left (\frac{q}{p+q}, \frac{p-q}{2(p+q)} \right ) \ge \frac{2\rho}{ n}$, and $\gamma_{mix} = \frac{2\rho}{ n}$. With these parameters
 we set, for sufficiently small $c_1$ and large $c_2,c_3$,
\begin{align*}
\hspace{-.5em}\eta =\frac{c_1 \epsilon^2 \delta^2 \cdot \rho}{\log^3 \left (\frac{n}{\epsilon \delta \rho} \right )},\hspace{.5em}\hspace{.5em}T = \frac{c_2 \cdot n \cdot \left (\log^3 \left (\frac{n}{\epsilon \delta \rho} \right ) + \frac{\log \left (\frac{n}{\epsilon \delta \rho} \right )}{\epsilon}\right)}{\epsilon^2 \delta^2 \rho^2},\hspace{.5em} T' = \frac{c_3\cdot n \cdot \left (\log \left (\frac{n}{\epsilon \delta \rho} \right ) + \frac{1}{\epsilon} \right )}{\rho^2}
\end{align*}
where to bound $T'$ we use that $\frac{\lambda_1(\bv{D}+\bv{W})}{\bgap} \le \frac{2}{\rho}$.
Let $\bv{\hat V} \in \R^{n \times k}$ be given by $(\bv{\hat V})_{u,j} = {\hat v}_u^{(j)}$ where ${\hat v}_u^{(j)}$ are the states of $\texttt{AsynchOja}(T,T',\eta)$ and let $\bv{\hat v}_2$ be the second column of $\bv{\hat V}$. With these parameters, \autoref{cor:ojaAsync2} gives with probability  $\ge 1-\delta$ that
$\left |\bv{\hat v}_2^T \bv{v}_{2}\right| \ge 1-\epsilon$ and $\norm{\bv{\hat v}_2}_2 \le 1+\epsilon.$

 Applying \autoref{approx_eig_starting_point} then gives the theorem if we adjust $\epsilon$ by a factor of $1/5$. Recall that the second eigenvector of $\bv{D}+\bv{W}$ is identical to that of $\bv{C}^{(p,q)}$. Additionally, in expectation, each node is involved in $L = \frac{2(T+T')}{n}$ interactions. This bound holds for all nodes within a factor $2$ with probability $1-\delta$ by a Chernoff bound, since $L = \Omega(\log(n/\delta))$. We can union bound over our two failure probabilities and adjust $\delta$ by $1/2$ to obtain overall failure probability $\le \delta$.
 \end{proof}

\subsection{Community Detection in the \modelg}

In the \modelg, nodes communicate using a random graph which is equal to the communication graph in the \modelpq\  \emph{in expectation}. Using an approach similar to \cite{spielmanNotes}, which is a simplifies the perturbation method used in \cite{959929}, we can prove that in the \modelg\ $\bv{W}$ is a small perturbation of $\bv{C}^{(p,q)}$ and so the second eigenvector of $\bv{D}+\bv{W}$ approximates that of $\bv{C}^{(p,q)}$ -- i.e., the cluster indicator vector $\bs{\chi}$.
We defer this analysis to \autoref{app:gnp}, stating the main result here:

\begin{shaded}
\vspace{-.75em}
\begin{theorem}[$\epsilon$-approximate community  detection: \modelg]\label{thm:modelg}
Consider \ripref{alg:commDetec} in the \modelg. Let $\rho = \min \left (\frac{q}{p+q}, \frac{p-q}{p+q} \right )$. For sufficiently small constant $c_1$ and sufficiently large $c_2$ and $c_3$ let
\begin{align*}
\hspace{-.5em}\eta =\frac{c_1 \epsilon^2 \delta^2 \rho}{\log^3 \left (\frac{n}{\epsilon \delta \rho} \right )},\hspace{.5em}T = \frac{c_2 n \left (\log^3 \left (\frac{n}{\epsilon \delta \rho} \right ) + \frac{\log \left (\frac{n}{\epsilon \delta \rho} \right )}{\epsilon}\right)}{\epsilon^2 \delta^2 \rho^2},\hspace{.5em}T' = \frac{c_3 n \left (\log \left (\frac{n}{\epsilon \delta \rho} \right ) + \frac{1}{\epsilon} \right )}{\rho^2}.
\end{align*}
If $\frac{\min\left  [q, p-q \right]}{\sqrt{p+q}} \ge \frac{c_4\sqrt{\log (n/\delta)}}{\epsilon \sqrt{n}}$ for large enough constant $c_4$,
then, with probability  $1-\delta$, after ignoring $\epsilon n$ nodes, all remaining nodes in $V_1$ terminate in some state $s_1 \in \{-1,1\}$, and all  nodes in $V_2$ terminate in state $s_2 = -s_1$. Supressing polylogarithmic factors, the total number of global rounds and local rounds required are: $T+T' = \tilde O \left ( \frac{n}{\epsilon^3 \delta^2 \rho^2} \right )$ and $L = \tilde O \left (\frac{1}{\epsilon^3 \delta^3  \rho^2} \right)$.
\end{theorem}
\vspace{-.75em}
\end{shaded}
If for example, $p,q = \Theta(1)$ and thus the $G(n,p,q)$  graph is dense, we can recover  the communities with probability $1-\delta$ up to $O(1)$ error as long as $q \le p - c\sqrt{\log (n/\delta)/n}$ for sufficiently  large constant $c$. Alternatively, if $p,q = \Theta \left  (\log(n/\delta)/n\right  )$, so the $G(n,p,q)$  graph is sparse, we require $q \le cp$ for sufficiently small $c$.

\section{Cleanup Phase}\label{sec:cleanup}
After we apply \autoref{thm:modelg} (respectively, \autoref{thm:modelp}) an $\epsilon$-fraction of nodes are incorrectly clustered.
The goal of this section is to provide a simple algorithm that improves this clustering so that \emph{all nodes} are labeled correctly after a small number of rounds.

For the \modelpq, doing so is straightforward. After running \ripref{alg:ojaGossip} and selecting a label, each time a node communicates in the future it records the chosen label of the node it communicates with. Ultimately, it changes its label to the majority of labels encountered. If $\epsilon$ is small enough so $p(1-\epsilon) > q+\epsilon p$, this majority tends towards the node's correct label.
\ifArxiv
\footnote{This bound might look crude, but it is in fact the strongest we can assume, since our techniques 
only give an upper bound on the fraction of incorrectly labeled nodes per clusters and it is conceivable that all nodes of the other cluster are already correctly labeled in which case the bound is tight.}
\fi
The number of required rounds for the majority to be correct, with good probability for all nodes, is a simple a function of $p,q$, and $\epsilon$.

The \modelg is more difficult. \autoref{thm:modelg} does not guarantee how incorrectly labeled nodes are distributed: it is possible that 
a majority of a node's neighbors fall into the set of $\epsilon n$ ``bad nodes''. In that case, even after infinitely many rounds of communication, the majority label encountered will not tend towards the node's correct identity. 

As a remedy, we introduce a phased algorithm (\ripref{alg:clean}) where each node updates its label to the majority of labels seen during a phase. We show that in each phase the fraction of incorrectly labeled nodes decreases by a constant factor. Our analysis establishes a graph theoretic bound on the external edge density of most subsets of nodes. Specifically, for all subsets $S$ below a certain size, 
we show that, with high probability, there are at most $|S|/3$ nodes which have enough connections to $S$ so that if an adversary gave all nodes in $S$ incorrect labels, it could cause these nodes to have an incorrect majority label. This bound guarantees that at most $|S|/3$ bad labels `propagate' to the next phase of the algorithm.
\medskip
\ifArxiv
Our property guarantees that $2|S|/3$ nodes will 
have a neighborhood in which a significant fraction of the nodes are correctly labeled. Hence, taking enough labels and taking Union bound over all nodes guarantees that the set of incorrectly nodes decreases by a constant factor w.h.p..
\fi

We analyze \ripref{alg:clean} in the \modelpq in \autoref{sec:cleanpq} (in this case we just set $k =1$) and \modelg in \autoref{sec:cleang}.
\begin{algorithm}[H]
\caption{\algoname{Cleanup phase} (pseudocode for node $u$) \newline
{\bf Input:} Number of phases $k$ and number of rounds per phase $r$.\newline
{\bf Output:} Label ${\hat \chi}_u \in \{-1,1\}$ 
}
\begin{algorithmic}[1]
\For{Phase $1$ to $k$}
\For{Round $i=1$ to $r$}
\State{$S_i := {\hat \chi}_v$, where $ {\hat \chi}_v$ denotes the $i^\text{th}$ sample of node  $u$. }
\EndFor

\State{${\hat \chi}_u:=1$ if $\sum_i^r S_i \geq 0$, ${\hat \chi}_u:=-1$ otherwise.}
\EndFor

\end{algorithmic}
\label{alg:clean}
\end{algorithm}

\begin{shaded}
\vspace{-.75em}
\begin{theorem}\label{thm:cleanuppq}
Consider the \modelpq.
Assume that a fraction of at most $\epsilon \leq 1/64$ of the nodes are incorrectly clustered after \ripref{alg:ojaGossip}. As long as $p'=(1-\epsilon)p$ and $q'=q+\epsilon p$ satisfy $p' > q'$, \ripref{alg:clean} ensures that all nodes are correctly labeled with high probability after
 $O(\frac{p\ln n}{  (\sqrt{p'}-\sqrt{q'})^2 })$ local rounds.
In particular, for $q\leq p/2$ and $\epsilon < 1/8$, the number of local rounds required is $O( \log n)$.
\end{theorem}
\vspace{-.75em}
\end{shaded}

\begin{shaded}
\vspace{-.75em}
\begin{theorem}\label{thm:cleanupg}
Consider the \modelg.
Let $\Delta = \frac{p}{2}-\frac{q}{2}-\sqrt{12p\ln n/n}-\sqrt{12q\ln n/n}$.  
Assume that $\Delta =\Omega(\ln n /n)$ and at most $\epsilon \leq \Delta/24p$ nodes are incorrectly clustered after \ripref{alg:ojaGossip}. As long as $p''= \frac{p}{2}-\sqrt{\frac{6 p\ln n}{n}}  -\frac{\Delta}{12}$ and $q''=\frac{q}{2} + \sqrt{\frac{6 q\ln n}{n}}  +\frac{\Delta}{12}$ satisfy $p'' > q''$, 
\ripref{alg:clean} ensures that all nodes are correctly labeled with high probability after
 $O(\frac{p\ln^2 n}{  (\sqrt{p''}-\sqrt{q''})^2 })$ local rounds.
In particular, for $q\leq p/2$ the number of local rounds required is $O( \log^2 n)$.
\end{theorem}
\vspace{-.75em}
\end{shaded}
Note that if $p-q =\Omega(\sqrt{\log n/n})$, then  $\Delta$ simplifies to $\Delta=\Theta(p-q)$. 
Incidentally,  $p-q =\Omega(\sqrt{\log n/n})$ is sometimes tight because, in this regime, clustering correctly can be infeasible: some nodes will simply have more neighbors in the opposite cluster. Consider for example when $p=1/2+ \sqrt{\ln n/(10n)}$ and $q=1/2$. 
%
\bibliographystyle{plainurl}

\bibliography{distributedClustering}

\appendix

\section{Oja's Error Bound}\label{app:oja}

In this section we give a full proof of \autoref{thm:oja}, restated below:
\begin{reptheorem}{thm:oja}
Let $\bv{M} \in\mathbb{R}^{n \times n}$ be a PSD matrix satisfying: $\frac{\sum_{i=1}^k \lambda_i(\bv{M})}{C} \le \Lambda$, and $\frac{\lambda_k(\bv{M}) - \lambda_{k+1}(\bv{M})}{C} \ge \gap$ for some bounds $\Lambda,\gap$. For any $\epsilon,\delta \in (0,1]$, let $\xi \eqdef \frac{n}{\delta \epsilon \cdot \gap}$, $\eta = \frac{c_1\epsilon^2 \cdot \gap \cdot \delta^2}{C\Lambda k \log^3 \xi}$ for sufficiently small $c_1,$ and $T = \frac{c_2\cdot(\log \xi+1/\epsilon)}{C\cdot\gap \cdot \eta}$ 
for sufficiently large $c_2$. Then \ripref{alg:oja} run with step size $\eta$ satisfies, with probability $\ge 1 -\delta$, for all $t \in [T-\frac{c_2/\epsilon}{2C\gap \cdot \eta},T+\frac{c_2/\epsilon}{2C\gap \cdot \eta}]$, letting $\bv{\tilde V}_t = \orth(\bv{Q}_t)$:
 $$\norm{\bv{Z}^T \bv{\tilde V}_t}_F^2 \le \epsilon.$$
 where $\bv{Z}$ is an orthonormal basis for the bottom $n-k$ eigenvectors of $\bv{M}$.
\end{reptheorem}

\begin{proof}
We first note that Theorem 1 of \cite{allen2016first} requires $\Pr_{\bv{x} \sim \dist}[\norm{\bv{x}}_2^2 \le 1] = 1$, while \autoref{thm:oja} allows vectors with norm up to some bound $C$. It is clear that this suffices since we scale the step size $\eta$, along with the $\gap$ and $\Lambda$ parameters by a factor of $\frac{1}{C}$ as compared to their definitions in \cite{allen2016first}. This translates to applying Theorem 1 to $\frac{1}{C} \bv{M}$ -- which has identical eigenvectors of $\bv{M}$. Thus, for the remainder of the proof we relable $\bv{M} = \frac{1}{C} \bv{M}$ and $\eta = C \eta$. 

We next note that 
Theorem 1 of \cite{allen2016first} sets $\Lambda = \sum_{i=1}^k \lambda_i(\bv{M})$, and $\gap = \lambda_k(\bv{M}) - \lambda_{k+1}(\bv{M})$, while in \autoref{thm:oja} we just require $\Lambda$ and $\gap$ to be bounds on the respective quantities. It is not hard to see that the theorem still holds with these  bounds. Specifically, the proof of Theorem 1 follows from the proof of Theorem 2 with parameter $\rho$ set to any value $\gap \le \lambda_k(\bv{M}) - \lambda_{k+1}(\bv{M})$. $\Lambda$ is also only used as an upper bound on $\sum_{i=1}^k \lambda_i(\bv{M})$ in the proof of Lemma iii.K.1 and the subsequent proofs of Lemma Main 4 and Lemma Main 5. 

Finally, we must argue that Theorem 5 follows from the proof of Theorem 1, even though our version of Oja's algorithm uses a fixed step size $\eta$, rather than a step size which changes by round. 

Even in our fixed $\eta$ setting we employ the analysis of \cite{allen2016first}, which considers three epochs of rounds.
Let $\xi \eqdef \frac{n}{\delta \epsilon \cdot \gap}$.
Set $\eta = \frac{c_1\epsilon^2  \cdot \gap \cdot \delta^2}{\Lambda k \log^3 \xi}$,
$T_0 = \frac{c_2\log \xi}{\gap \cdot \eta}$, $T_1 = \frac{c_3}{\gap \cdot \eta}$, and $T_2 = \frac{c_4 T_1}{\epsilon}$ for some constants $c_1,c_2,c_3,c_4$. Overall we have $T = T_0 + T_1 + T_2 =  O \left (\frac{\Lambda k \log^3 \xi \cdot (\log \xi + 1/\epsilon)}{\epsilon^2 \cdot \gap^2 \cdot \delta^2} \right )$.

Theorem 1 of \cite{allen2016first} follows from Theorem 2. The proof of this theorem first invokes Lemma Main 4. For sufficiently  small $q = \poly(\xi)$, $\Xi_{\bv{Z}} = \Theta \left ( \frac{nk}{\delta^2} \ln \frac{n}{\delta} \right )$ and $\Xi_x = \Theta \left ( \frac{\sqrt{k \ln \frac{T}{\delta}}}{\delta} \right )$, letting $\eta$ be a step size which is fixed over all rounds, it requires:
\begin{align}
\frac{2q(\Xi_{\bv{Z}}^{3/2} + n \Xi_x^2)}{\Lambda} \le \eta = O \left (\frac{\gap}{\Lambda \Xi_x^2 \ln \frac{nT}{q}} \right )\label{req1}\\
\sum_{t=1}^{T} \Lambda \eta^2 \Xi_x^2 = O \left (\frac{1}{\ln \frac{nT}{q}} \right )\label{req2}\\
\exists T_0 \le T: \sum_{t=1}^{T_0} \eta = \Omega \left ( \frac{\ln \Xi_{\bv{Z}}}{\gap} \right)\label{req3}.
\end{align}
Using that $\Xi_x = \Theta \left ( \frac{\sqrt{k \ln \frac{T}{\delta}}}{\delta} \right )$, the upper bound of
\eqref{req1} holds as long as $\eta = O\left( \frac{\gap \cdot \delta^2}{\Lambda k \log^2 \xi} \right)$, which is satisfied by our setting of $\eta$ if $c_1$ is small enough. The lower bound holds easily as well since $q = \poly(\xi)$ for some sufficiently small polynomial.

 \eqref{req2} holds as long as $\eta = O \left (\frac{\delta}{\sqrt{\Lambda T k} \log \xi} \right )$, which holds by our setting of $\eta$ and $T$ as long as $c_1$ is sufficiently small compared to $c_2$, $c_3$, and $c_4$. Finally, \eqref{req3} holds if $T_0 = \Omega \left (\frac{\log \frac{nk}{\delta}}{\gap \cdot \eta} \right )$, which again holds for our setting of parameters.

In the proof of Theorem 2, Lemma Main 6 is next invoked with $\Xi_x = \Xi_{\bv{Z}} = 2$ and $T_0$ identified as the first round of computation. This lemma requires the conditions of Lemma Main 4 (discussed above) along with those of Lemma Main 5, which require that there is some $\Delta \le 1/\sqrt{8}$ such that:
\begin{align}
\frac{T_1}{\ln^2 T_1} = \Omega \left ( \frac{\log \frac{n}{q}}{\Delta^2} \right )\label{req4}\\
\forall t \in [T_1+1,T]: 2\eta \gap - 4\eta^2 = \Omega \left ( \frac{1}{t} \right )\label{req5}\\
\eta = O \left ( \frac{1}{\sqrt{\Lambda} t \Delta} \right )\label{req6}.
\end{align}

In fact, for our result, we can invoke Lemma Main 5 directly. Since we use an eigengap assumption, this lemma bounds (in the notation of \cite{allen2016first}) $\norm{\bv{Z}^T \bv{P}_t \bv{Q}(\bv{V}^T \bv{P}_t \bv{Q})^{-1}}_F^2 \le \frac{5T_1/\ln^2(T_1)}{(t-T_0)/\ln^2 (t-T_0)}$, for all $t \in [T_0+T_1,T]$. This upper bounds $\norm{\bv{Z}^T\bv{\tilde V}_t}_F^2$ by Lemma 2.2. Since we set $T_2 = \frac{c_4 T_1}{\epsilon}$ and $T_1 = \frac{c_3}{\gap \cdot \eta} \ge \frac{1}{\epsilon}$, setting $c_4$ sufficiently large gives $\norm{\bv{Z}^T \bv{\tilde V}_T}_F^2 \le \epsilon$ and in fact, $\norm{\bv{Z}^T \bv{\tilde V}_t}_F^2 \le \epsilon$ for all $t \in [T- (1-c_5)T_2,T]$ for any constant $c_5 \in (0,1]$. Our final error bound follows by driving $c_5$ sufficiently small and noting that we can shift $T$, multiplying it by at most constant factor, so that our bound holds for all $t \in [T- (1-c_5) T_2,T+(1-c_5) T_2]$.

It just remains to verify the conditions of Lemma Main 5. We set $\Delta = \frac{c_5\gap \epsilon}{\sqrt{\Lambda}}$ for sufficiently small $c_5$. 
 \eqref{req4} holds since $T_1 = \frac{c_3}{\gap \eta} = \Omega \left (\frac{\log \xi \cdot \Lambda}{\gap^2 \cdot \epsilon^2} \right )$. For \eqref{req5}, first note that if $c_1$ is set small enough, we can lower bound $2\eta \gap - 4 \eta^2 \ge \eta \gap$. Thus \eqref{req5} holds since we have $T_1 = \frac{c_3}{\gap \cdot \eta}$ for sufficiently large $c_3$. Finally, \eqref{req6} holds since we have $\sqrt{\Lambda} T_2 \Delta = c_5c_4 T_1 \gap = \frac{c_3c_4c_5}{ \eta}$. Thus, all conditions of Lemma Main 5 hold with our fixed $\eta$, yielding the theorem.
\end{proof}

\section{Asynchronous Averaging}\label{app:averaging}

In this section we introduce a simple asynchronous averaging algorithm, which is used in our orthogonalization routine \texttt{AsynchOrth} (\ripref{alg:pporth}), analyzed in \autoref{app:orth}.

Our bounds closely follow classic work on asynchronous gossip algorithms \cite{boyd2005gossip}, however we include a full proof for completeness, since our setting is more general than typically considered.

\begin{algorithm}[H]
\caption{\algoname{Asynchronous Averaging}}
{\bf Initialization}: Each $u$ holds value $x_u$. $y_u := x_u$.\\
{\bf Update}: If $(u,v)$ is chosen by the randomized scheduler:
\begin{algorithmic}[1]
\State{$y_u := \frac{y_u + y_v}{2}.$}
\end{algorithmic}
\label{alg:ppaverage}
\end{algorithm}
\begin{lemma}\label{lem:avg} Consider a set of nodes executing \ripref{alg:ppaverage} in the asynchronous communication model (\autoref{def:weightedPopProtocol}) with weight matrix $\bv{W}$ and degree matrix $\bv{D}$. Let $x_{avg} = \frac{1}{n}\sum_{v \in V} x_v$. With probability $\ge 1-\delta$ in all rounds $t \ge \frac{\log (1/\epsilon\delta)}{\log \left(\lambda_2^{-1}(\bv{I}-\frac{1}{2}\bv{D} + \frac{1}{2}\bv{W})\right)}$, $\sum_{v \in V} (y_v - x_{avg})^2 \le \epsilon \cdot \sum_{v \in V} (x_v - x_{avg})^2$.
\end{lemma}
That is, the mean squared error of the estimates of $x_{avg}$ converges linearly, with rate dependent on the second eigenvalue of $\bv{I}-\frac{1}{2}\bv{D} + \frac{1}{2}\bv{W}$. 
Note that $\bv{D} - \bv{W}$ is the Laplacian matrix corresponding to the communication model and $\log \left(\lambda_2^{-1}(\bv{I}-\frac{1}{2}\bv{D} + \frac{1}{2}\bv{W})\right) = \log \left (\frac{1}{1 - \frac{1}{2}\lambda_{n-1}(\bv{L})} \right )\approx \lambda_{n-1}(\bv{L})$ is roughly its smallest nonzero eigenvalue.
\begin{proof}
Our proof follows that of \cite{boyd2005gossip,boyd2006randomized}. We can write $x_{avg} = \frac{1}{n}\sum_{v \in V} x_v = \frac{1}{n} \cdot \bv{1}^T\bv{x}$ where $\bv{x} \in \R^n$ contains each node's value as its entries. Recall that for convenience we identify the vertex set $V$ with $[n] = \{1,...,n\}$. So the $u^{th}$ entry of $\bv x$ contains $x_u$. Let $\bv{y}^t$ denote the vector of estimated averages at round $t$. Initially $\bv{y}^0 = \bv{x}$.

In each step of \ripref{alg:ppaverage}, two nodes $u$ and $v$, selected with probability $\bv{W}_{u,v}$, average their two values. We can write this update as a matrix product with $\bv{y}^t$. Specifically, for any pair $(u,v)$, let $\euv = \bv{e}_u - \bv{e}_v$. We have, for $\euv$ chosen with probability $\bv{W}_{u,v}$, 
\begin{align}
\bv{y}^{t+1} = (\bv{I} - \frac{1}{2}\euv \euv^T) \bv{y}^t.\label{eq:avgUpdate}
\end{align}
We first note that $\bv{1}^T(\euv \euv^T) = \bv{0}$. Thus, by \eqref{eq:avgUpdate}, for every $t$, we have $\frac{1}{n}\bv{1}^T \bv{y}^t = \frac{1}{n}\bv{1}^T \bv{y}^0 = \frac{1}{n} \bv{1}^T\bv{x} = x_{avg}$. That is, the average value held at the nodes always equals the true average. We bound the error from this average: $\bv{z}^t \eqdef \bv{y}^t - x_{avg} \cdot \bv{1} = (\bv{I} - \frac{1}{n}\bv{1}\bv{1}^T) \bv{y}^t$. 
By \eqref{eq:avgUpdate} we have:
\begin{align*}
\bv{z}^{t+1} = (\bv{I} - \frac{1}{n}\bv{1}\bv{1}^T) (\bv{I} - \frac{1}{2}\euv \euv^T) \bv{y}^{t} &=(\bv{I} - \frac{1}{2}\euv \euv^T)(\bv{I} - \frac{1}{n}\bv{1}\bv{1}^T) \bv{y}^t\nonumber\\
&= (\bv{I} - \frac{1}{2}\euv \euv^T) \bv{z}^t
\end{align*}
where the second step follows since $\euv \euv^T \bv{1}\bv{1}^T = \bv{1}\bv{1}^T \euv \euv^T = \bv{0}$ and $\bv{I}$ commutes with all matrices. We can thus compute the expected norm $\norm{\bv{z}^{t+1}}_2^2 = (\bv{z}^{t+1})^T \bv{z}^{t+1}$ as:
\begin{align}\label{eq:initalExpectation}
\E \left [(\bv{z}^{t+1})^T \bv{z}^{t+1} \big | \bv{z}^t \right] = \sum_{\pairs}\bv{W}_{u,v} \cdot (\bv{z}^t)^T \left(\bv{I} - \frac{1}{2}\euv \euv^T\right)^2 \bv{z}^t
\end{align}
recalling that $\pairs$ is the set of unordered pairs $(u,v)$ with $u \neq v$ and $\bv{W}_{u,v}$ is the probability that such a pair is chosen to interact in any round $t$.
$\left(\bv{I} - \frac{1}{2}\euv \euv^T\right)^2 = \bv{I} - \euv \euv^T + \frac{1}{4} (\euv \euv^T)^2 = \bv{I} - \frac{1}{2}\euv \euv^T$ since $(\euv \euv^T)^2 = 2 \euv \euv^T$. That is, $\bv{I} - \frac{1}{2}\euv \euv^T$ is a projection matrix. Plugging back into \eqref{eq:initalExpectation}:
\begin{align}
\E \left [(\bv{z}^{t+1})^T \bv{z}^{t+1} \big | \bv{z}^t \right] &= (\bv{z}^t)^T \left (\sum_{\pairs}\bv{W}_{u,v} \cdot  \left(\bv{I} - \frac{1}{2}\euv \euv^T\right) \right) \bv{z}^t\nonumber\\
&= (\bv{z}^t)^T \left (\bv{I} - \frac{1}{2}(\bv{D}-\bv{W}) \right ) \bv{z}^t\label{eq:expectation2}.
\end{align}
Denote $\bv{H} \eqdef (\bv{I} - \frac{1}{2}(\bv{D}-\bv{W}))$. $\bv{H}$ is a sum over PSD matrices $\left (\bv{I} - \frac{1}{2}\euv \euv^T\right)$ and so is itself PSD. It has top eigenvalue $\lambda_1(\bv{H}) = 1$, and top eigenvector $\bv{1}$ (we can see this since $(\bv{D}-\bv{W})\bv{1} = \bv{0}$). Since $\bv{1}^T \bv{z}^t = \bv{1}^T (\bv{I}-\frac{1}{n}\bv{1}\bv{1}^T)\bv{y}^t = 0$, we thus have from \eqref{eq:expectation2}, $\E \left [\norm{\bv{z}^{t+1}}_2^2\ \big |\ \norm{\bv{z}^t}_2^2 \right ] \le \lambda_2(\bv{H}) \norm{\bv{z}^t}_2^2$ and by iterating:
\begin{align}\label{eq:avgConvergence}
\E \left [\norm{\bv{z}^{t}}_2^2 \right ] \le \lambda_2(\bv{H})^t \norm{\bv{z}^0}_2^2
\end{align} 
If we set $t \ge \frac{\log(1/\epsilon \delta)}{\log \left(\lambda_2(\bv{H})^{-1}\right)}$ applying \eqref{eq:avgConvergence} gives $\E \left [\norm{\bv{z}^{t}}_2^2 \right ] \le \epsilon \delta \norm{\bv{z}^0}_2^2$ and thus by Markov's inequality, $\Pr \left [\norm{\bv{z}^{t}}_2^2 \le \epsilon \norm{\bv{z}^0}_2^2 \right ] \ge 1 - \delta$. This gives the bound since $\norm{\bv{z}^t}_2^2 = \sum_{v \in V} (y_v^t - x_{avg})^2$. It just remains to note that since $\bv{z}^{t+1} = \left ( \bv{I} - \frac{1}{2}\euv \euv^T\right ) \bv{z}^t$ and $\norm{\bv{I} - \frac{1}{2}\euv \euv^T}_2 = 1$, the error strictly decreases in each round, so once it is bounded in round $t$, it is bounded in all subsequent rounds.
\end{proof}

\section{Distributed orthogonalization proofs}\label{app:orth}
In this section we give a full analysis of the distributed orthogonalization routine \texttt{AsynchOrth} described in \ripref{alg:pporth}. We are interested in the error bounds it gives for eigenvector approximation when used as a subroutine in \ripref{alg:ojaGossip}. We begin by analyzing the idealized case when \ripref{alg:pporth} is assumed to \emph{exactly} implement the centralized \ripref{alg:chol}. We then account for the fact that the distributed implementation is approximate.

\begin{corollary}[Distributed eigenvector approximation with exact Cholesky orthogonalization]\label{cor:ojaAsync}
Let $\bv{v}_1,...,\bv{v}_k$ be the top $k$ eigenvectors of the communication matrix $\bv{D}+\bv{W}$ in an asynchronous communication model, and let $\Lambda,\bgap,$ be  bounds satisfying: $\sum_{j=1}^k \lambda_j(\bv{D}+\bv{W}) \le \Lambda$ and $\min_{j \in [k]} [\lambda_j(\bv{D}+\bv{W})-\lambda_{j+1}(\bv{D}+\bv{W})] \ge \bgap$.

For any $\epsilon,\delta \in(0,1)$, let $\xi = \frac{n}{\delta \epsilon \cdot \bgap}$ and let $\eta = \frac{c_1 \epsilon^2 \cdot \bgap \cdot \delta^2}{\Lambda k^3 \log^3 \xi}$ for sufficiently small $c_1$ and $T = \frac{c_2 \cdot (\log \xi + 1/\epsilon)}{\bgap \cdot \eta}$ for sufficiently large $c_2$.
For all $u \in [n],i \in [k]$, let ${q}_u^{(j)}$ be the local state computed by \ripref{alg:ojaGossip} prior to Step 4. 
Then for all $u \in [n],i \in [k]$, let ${\tilde v}_u^{(j)}$ be result of running \texttt{AsynchOrth} in Step 4 with an algorithm that exactly implements \ripref{alg:chol}. If $\bv{\tilde V}$ is given by $\left(\bv{\tilde V}\right)_{u,j}  = {\tilde v}_u^{(j)}$ and $\bv{\tilde v}_i$ is the $i^\text{th}$ column of $\bv{\tilde V}$, then with probability  $\ge 1 -\delta$:
 \begin{align*} \left |\bv{\tilde v}_i^T \bv{v}_{i} \right | \ge 1-\epsilon\hspace{1em}\text{ and }\hspace{1em}\norm{\bv{\tilde v}_i}_2 = 1\hspace{1em}\text{ for all } i \in [k].
 \end{align*}
\end{corollary}
We note that a similar eigenvector bound is given in \cite{allen2016first} with better dependence on $k$. However, it has worse dependence on $\Lambda$ and so is too weak for our applications.
\begin{proof}
Let $\bv{\tilde Q}$ be given by $\left(\bv{Q}\right)_{u,j}  = {q}_u^{(j)}$.
\ripref{alg:chol} first computes $\bv{L} = \chol(\bv{Q}^T \bv{Q})$ and then $\bv{\tilde V} = \bv{Q} (\bv{L}^T)^{-1}$. Letting $\bv{Q}_i \in \R^{n \times i}$ denote the first $i$ columns of $\bv{Q}$, it is well known that, letting $\bv{L}_i$ denote the upper left $i \times i$ submatrix of $\bv{L}$, $\bv{L}_i = \chol(\bv{Q}_i^T \bv{Q}_i)$. Furthermore, since $\bv{L}$ is triangular, $(\bv{L}_i^T)^{-1}$ is just the upper $i \times i$ submatrix of  $(\bv{L}^T)^{-1}$. Thus we can see that the first $i$ columns of $\bv{V}$ are \emph{identical} to the output that would be
%
%
%
obtained if the algorithm were run with the same step size $\eta$ and step count $T$, but with just the first $i$ vectors of $\bv{Q}$ -- i.e. with each node only keeping track of just $q_u^{(1)},...,q_u^{(i)}$ instead of $q_u^{(1)},...,q_u^{(k)}$ 

With this observation, we can prove the corollary by applying \autoref{thm:oja} for each $i \in [k]$. 
We apply the theorem with rank $i$, error $\epsilon/2$, failure probability $\delta/k$, $\Lambda \ge \sum_{j=1}^k \lambda_j(\bv{D}+\bv{W}) \ge \sum_{j=1}^i \lambda_j(\bv{D}+\bv{W})$, and $\bgap \le \min_{j\in[k]} [\lambda_j(\bv{M}) - \lambda_{j+1}(\bv{M})] \le \lambda_i(\bv{D}+\bv{W}) - \lambda_{i+1}(\bv{D}+\bv{W})$. 

Denote $\bv{V}_{i} = [\bv{v}_1,...,\bv{v}_i]$, and let $\bv{V}_{-i}$ span the remaining eigenvectors of $\bv{D} +\bv{W}$. 
By our application of the theorem with rank $i$, with probability  $\ge 1-\delta/k$, $\norm{\bv{V}_{-i}^T\bv{\tilde V}_i}_F^2 \le \epsilon/2$ and:
\begin{align}
\norm{\bv{V}_{i}^T\bv{\tilde V}_i}^2_F \ge i-\epsilon/2\label{eq:pythagBound}
\end{align}
 since $\norm{\bv{\tilde V}_i}_F^2 = i$ and by the Pythagorean theorem, $\norm{\bv{\tilde V}_i}_F^2 = \norm{\bv{V}_{i}^T\bv{\tilde V}_i}^2_F + \norm{\bv{V}_{-i}^T\bv{\tilde V}_i}^2_F$. This holds for all $i$ simultaneously with probability $\ge 1-\delta$ after union bounding over $k$ applications of the theorem.
 Since $\norm{\bv{V}_i^T \bv{\tilde V}_{i-1}}_F^2 \le \norm{\bv{\tilde V}_{i-1}}_F^2 = i-1$, \eqref{eq:pythagBound} gives: 
\begin{align}\label{induct1}
\norm{\bv{V}_i^T \bv{\tilde v}_{i}}_2^2 = \norm{\bv{V}_{i}^T\bv{\tilde V}_i}_F^2- \norm{\bv{V}_i^T \bv{\tilde V}_{i-1}}_F^2 \ge 1-\epsilon/2.
\end{align}
For $i = 1$, this completes the proof since $\norm{\bv{V}_1^T \bv{\tilde v}_{1}}_2^2 = (\bv{v}_1^T \bv{\tilde v}_{1})^2$. For $i > 1$, we also have by our application of \autoref{thm:oja} with rank $i-1$, $\norm{\bv{V}_{i-1}^T \bv{\tilde V}_{i-1}}_F^2 \ge i-1-\epsilon/2$ and so:
\begin{align}\label{induct2}
\norm{\bv{V}_{i-1}^T \bv{\tilde v}_i}_2^2 \le \norm{\bv{\tilde V}_{i-1}}_F^2 - \norm{\bv{V}_{i-1}^T \bv{\tilde V}_{i-1}}_F^2 \le \epsilon/2.
\end{align}
Since $\norm{\bv{V}_i^T \bv{\tilde v}_{i}}_2^2 = \norm{\bv{V}_{i-1}^T \bv{\tilde v}_i}_2^2 +(\bv{v}_i^T \bv{\tilde v}_i)^2$, in combination \eqref{induct1} and \eqref{induct2} give $(\bv{v}_i^T \bv{\tilde v}_i)^2 \ge 1-\epsilon$ and hence $\left |\bv{v}_i^T \bv{\tilde v}_i \right | \ge 1-\epsilon$. $\norm{\bv{\tilde v}_i}_2 = 1$ follows from: $$\bv{\tilde V}^T \bv{\tilde V} = \bv{L}^{-1} \bv{Q}^T \bv{Q} (\bv{L}^T)^{-1} = \bv{L}^{-1} \bv{L} \bv{L}^T (\bv{L}^T)^{-1} = \bv{I}_{k \times k}.$$
\end{proof}

With \autoref{cor:ojaAsync} in place, we next focus on the additional error introduced by the fact that \ripref{alg:pporth} only implements \ripref{alg:chol} approximately. We first bound how well \ripref{alg:pporth}  approximates $\bv{ Q}^T\bv{ Q}$ via averaging.
\begin{lemma}\label{lem:dotProducts} Consider a set of nodes executing \texttt{AsynchOrth}$(T)$ (\ripref{alg:pporth}) in the asynchronous communication model with weight matrix $\bv{W}$ and degree matrix $\bv{D}$ satisfying: \\$\gamma_{mix} \le \min \left[\frac{1}{n}, \log \left(\lambda_2^{-1}(\bv{I}-\frac{1}{2}\bv{D} + \frac{1}{2}\bv{W})\right) \right]$. Let $\bv{Q} \in \R^{n \times k}$ be given by $(\bv{Q})_{u,j} = {q}_u^{(j)}$.

For any $\epsilon,\delta \in (0,1)$, if $T \ge \frac{c \log \left (\frac{n\norm{\bv{ Q}}_2}{\epsilon \delta}\right)}{\gamma_{mix}}$ for sufficiently large constant $c$, with probability $\ge 1-\delta$, for all $u$, and all $i,j \in [k]$, $$\left |(\bv{R}_u)_{i,j} - (\bv{ Q}^T \bv{ Q})_{i,j} \right | \le \epsilon.$$
\end{lemma}
\begin{proof}

We apply \autoref{lem:avg} of \autoref{app:averaging} to bound the accuracy of the averaging protocol in computing $(\bv{ Q}^T \bv{ Q})_{i,j}$. Specifically, with probability $\ge 1-\delta$, for every round $t' \ge T$, letting $\epsilon' =\left( \frac{\epsilon}{n^{3/2} \norm{\bv{Q}}_2}\right)^2$ the theorem gives:
\begin{align}\label{avgInitialBound}
\left [(\bv{R}_u)_{i,j}- \frac{1}{n}(\bv{ Q}^T \bv{ Q})_{i,j}\right]^2 \le n^2 \cdot \sum_{u \in V} \left [r_u^{(i,j)} \frac{1}{n}(\bv{ Q}^T \bv{ Q})_{i,j}\right]^2 \le n^2\epsilon' \cdot \sum_{u \in U} \left [{q}_u^{(i)} {q}_u^{(j)}-\frac{1}{n}(\bv{ Q}^T \bv{ Q})_{i,j} \right ]^2.
\end{align}
We can loosely bound:
\begin{align*}
\sum_{u \in U} \left [{q}_u^{(i)} {q}_u^{(j)}-\frac{1}{n}(\bv{ Q}^T \bv{ Q})_{i,j} \right ]^2 \le \sum_{u \in U} \left [{q}_u^{(i)} {q}_u^{(j)} \right ]^2 = \norm{\bv{Q}}_F^2 \le n\norm{\bv{ Q}}_2^2.
\end{align*}
The lemma follows by plugging into \eqref{avgInitialBound} and taking a square root of the error bound.
\end{proof}

\begin{remark} We note that, as shown in \autoref{lem:avg}, the accuracy of averaging to approximate $\bv{Q}^T\bv{Q}$ only decreases with each round. Thus, \texttt{AsynchOrth} (Algorithm  \ref{alg:pporth}) does not need to terminate the averaging after $T$ rounds, but can continuously maintain approximations to $\bv{Q}^T\bv{Q}$ and $[\hat v_u^{(1)},...,\hat v_u^{(k)}]$, which will only become more accurate over time.
\end{remark}

We next show how the entrywise approximation bound of \autoref{lem:dotProducts} translates into error in computing $\bv{\tilde V} = \bv{ Q} (\bv{L}_u^T)^{-1}$.  We use a forward stability result on the Cholesky decomposition:
\begin{theorem}[Theorem 10.8 of \cite{higham2002accuracy}, from \cite{sun1992rounding}]\label{thm:forward} Let $\bv{R} \in \R^{k\times k}$ be positive definite with Cholesky decomposition $\bv{R} = \bv{L}\bv{L}^T$. Let $\bs{\Delta}_{R}$ be a symmetric matrix satisfying $\norm{\bv{R}^{-1}\cdot \bs{\Delta}_{R}}_2 <1$. Then $\bv{R} + \bs{\Delta}_{R}$ has the Cholesky decomposition $\bv{R} + \bs{\Delta}_{R} = (\bv{L} + \bs{\Delta}_L)(\bv{L} + \bs{\Delta}_L)^T$ where:
\begin{align*}
\frac{\norm{\bs{\Delta}_L}_F}{\norm{\bv{L}}_2} \le \frac{\norm{\bv{R}^{-1}}_2 \cdot \norm{\bs{\Delta}_R}_F}{\sqrt{2}\left(1-\norm{\bv{R}^{-1}}_2 \cdot \norm{\bs{\Delta}_R}_F\right)}.
\end{align*}
\end{theorem}
Using this result we can show:
\begin{lemma}\label{lem:stabilityAnalysis}
Consider $\bv{Q} \in \R^{n \times k}$ with maximum and minimum singular values $\smax(\bv{Q})$ and $\smin(\bv{Q})$.
For all $u \in [n]$, let $\bv{R}_u \in \R^{k \times k}$ be any symmetric matrix with $\left |(\bv{R}_u)_{i,j} - (\bv{ Q}^T \bv{ Q})_{i,j} \right | \le \epsilon$ for all $i,j$ and some $\epsilon \le \frac{n\min(1,\smin(\bv{Q})^3)}{2 \sqrt{2} \max(1,\smax(\bv{Q}))}$. Let $\bv{L} = \chol(\bv{ Q}^T \bv{ Q})$, $\bv{L}_u = \chol(\bv{R}_u)$, $\bv{\tilde V} = \bv{ Q}(\bv{L}^T)^{-1}$,  and $\bv{\wh V}$ have $u^{th}$ row equal to the $u^{th}$ row of $\bv{ Q}(\bv{L}_u^T)^{-1}$. Then letting $\bv{\tilde v}_i$ and $\bv{\wh v}_i$ be the $i^{th}$ columns of $\bv{\tilde V}$ and $\bv{\wh V}$ respectively, for all $i \in [k]$, 
$$\norm{\bv{\hat v}_i - \bv{\tilde v}_i}_2 \le \frac{2\sqrt{2}\epsilon n^2 \smax(\bv{Q})}{\smin(\bv{Q})^3}.$$ 
\end{lemma}
\begin{proof}
Denote $\bv{R} = \bv{ Q}^T \bv{ Q}$. By the bound $\left |(\bv{R}_u)_{i,j} - (\bv{ Q}^T \bv{ Q})_{i,j} \right | \le \epsilon$ we can write $\bv{R}_u = \bv{R} + \bs{\Delta}_R$ where $\bs{\Delta}_R$ is symmetric with $\norm{\bs{\Delta}_R}_F \le n\epsilon$. We have
$$\norm{\bv{R}^{-1} \bs{\Delta}_R}_2 \le \norm{\bv{R}^{-1}}_2\cdot \norm{\bs{\Delta}_R}_F < \frac{n\epsilon}{\smin(\bv{Q})^2} < \frac{1}{2}$$
where the last step follows from our upper bound on $\epsilon$.
Plugging into \autoref{thm:forward} gives that $\bv{R}_u$ has Cholesky decomposition: $\bv{R}_u = \bv{L}_u \bv{L}_u^T$ where $\bv{L}_u = \bv{L} + \bs{\Delta}_L$ and
\begin{align}\label{lbound}
\norm{\bs{\Delta}_L}_2 \le \norm{\bs{\Delta}_L}_F \le \norm{\bv{L}}_2 \cdot \sqrt{2} \norm{\bv{R}^{-1}}_2 \norm{\bs{\Delta}_R}_F \cdot \le \frac{\sqrt{2}\cdot \epsilon n \cdot \smax(\bv{Q})}{\smin(\bv{Q})^2}.
\end{align}

Let $\bv{\tilde V}^{(u)}$ be the $u^{th}$ row of $\bv{\tilde V} = \bv{Q}(\bv{L}^T)^{-1}$ and $\bv{\wh V}^{(u)}$ be the $u^{th}$ row of $\bv{\wh V}$ equal to the $u^{th}$ row of $\bv{Q}(\bv{L}_u^T)^{-1}$. 
Using \eqref{lbound}, we can bound $\norm{\bv{L}^{-1}}_2 \cdot \norm{\bs{\Delta}_L}_2 \le \frac{\sqrt{2}\epsilon n \smax(\bv{Q})}{\smin(\bv{Q})^3} \le \frac{1}{2}$ by our upper bound on $\epsilon$. By a standard linear system stability bound (e.g., Theorem 7.2 of \cite{higham2002accuracy}):
\begin{align*}
\norm{\bv{\tilde V}^{(u)} - \bv{\wh V}^{(u)}}_2 &\le \norm{\bv{\tilde V}^{(u)}}_2 \cdot \frac{\norm{\bv{L}^{-1}}_2 \cdot \norm{\bs{\Delta}_L}_2}{1- \norm{\bv{L}^{-1}}_2 \cdot \norm{\bs{\Delta}_L}_2} \le \frac{2\sqrt{2}\epsilon n \smax(\bv{Q})}{\smin(\bv{Q})^3}
\end{align*}
where we bound $\norm{\bv{\tilde V}^{(u)}}_2 \le 1$ since $\bv{\tilde V}$ is orthonormal. Since every entry in $\bv{\tilde V}^{(u)} - \bv{\wh V}^{(u)}$ can be loosely bounded in magnitude by the vector's norm, summing over all $n$ rows gives $\norm{\bv{\hat v}_i - \bv{\tilde v}_i}_2 \le \frac{2\sqrt{2}\epsilon n^2 \smax(\bv{Q})}{\smin(\bv{Q})^3}$, completing the lemma.

\end{proof}
To apply \autoref{lem:stabilityAnalysis} we must bound $\smax(\bv{Q})$ and $\smin(\bv{Q})$. We have the following:
\begin{lemma}[Conditioning of $\bv{Q}$.]\label{lem:qcond}
Consider a set of nodes executing \texttt{AsynchOja}$(T,T',\eta)$ (\ripref{alg:ojaGossip}) in the asynchronous communication model with weight matrix $\bv{W}$ and with step size $\eta$ and stopping time $T$ as specified in \autoref{cor:ojaAsync}.  Let $\bv{ Q} \in\R^{n\times k}$ be given by $(\bv{ Q})_{u,j} = {q}_u^{(j)}$ after round $T$. For some large enough constant $c$, ith probability $1-\delta - e^{-\Theta(n)}$:
\begin{align*}
\smax(\bv{ Q}) \le e^{c(\log \xi + 1/\epsilon)\cdot \left [\frac{ \lambda_1(\bv{D}+\bv{W})}{\bgap}+1\right]} \hspace{1em}\text{ and }\hspace{1em} \smin(\bv{{Q}}) \ge \delta/c.
\end{align*}
\end{lemma}
\begin{proof} We can write $\bv{Q} = \left (\prod_{j=1}^T (\bv{I} + \eta \bv{A}_j) \right) \bv{Q}_0$ for $T = \frac{c_2(\log \xi + 1/\epsilon)}{\bgap \cdot \eta}$ where $\bv{A}_j = \euv \euv^T$ for some pair $u,v$. 
Recall that $\bv{Q}_0$ has all entries independently selected from $\mathcal{N}(0,1)$. We can loosely bound that, with probability $\ge 1-\delta/3$, $\norm{\bv{Q}_0}_2 \le \norm{\bv{Q}_0}_F \le cnk \sqrt{\log(1/\delta)}$ for some sufficiently large $c$. Additionally, by Theorem 1.1 of \cite{rudelson2009smallest} we have $\smin(\bv{Q}_0) \ge \delta/c$ with probability $\ge 1- \delta/3 - e^{-\Theta(n)}$ for sufficiently large $c$.
Starting with the lower bound, since $\bv{A}_j$ is always PSD and so all eigenvalues of $\bv{I}+\eta \bv{A}_j$ are $\ge 1$, 
\begin{align*}
\smin(\bv{ Q}) &\ge \prod_{j=1}^T \smin(\bv{I} + \eta \bv{A}_j) \cdot \smin(\norm{\bv{Q}_0}_2)\\
&\ge \smin(\norm{\bv{Q}_0}_2) \ge \delta/c.
\end{align*}

For the upper bound, for any  $0 \le t \le T$ let $\bv{B}_t = \prod_{j=1}^t (\bv{I} + \eta \bv{A}_j)$. Let $\bv{B}_0 = \bv{I}$. Define $\alpha_t  \eqdef \E[\tr(\bv{B}_t \bv{B}_t^T)] = \tr(\E[\bv{B}_t \bv{B}_t^T])$. We have $\alpha_t  \ge \E[\norm{\bv{B}_t}_2^2]$. Further, we can bound $\alpha_t$ by:
\begin{align*}
\alpha_t &= \E \left[\tr((\bv{I}+\eta \bv{A}_t)\bv{B}_{t-1} \bv{B}_{t-1}^T (\bv{I}+\eta \bv{A}_t)^T)\right]\\
&= \E  \left [\tr(\bv{B}_{t-1} \bv{B}_{t-1}^T (\bv{I}+\eta \bv{A}_t)^T(\bv{I}+\eta \bv{A}_t))\right ]\tag{By cyclic property of the trace}\\
&= \tr \left ( \E [\bv{B}_{t-1} \bv{B}_{t-1}^T] \E [(\bv{I}+\eta \bv{A}_t)^T(\bv{I}+\eta \bv{A}_t) ]\right )\tag{By independence of samples $\bv{A}_t$}\\
&= \tr \left ( \E [\bv{B}_{t-1} \bv{B}_{t-1}^T](\bv{I} + 2(\eta+\eta^2)(\bv{D}+\bv{W})) \right )\tag{Since $\E[\bv{A}_t] = \bv{D} + \bv{W}$ and $\bv{A}_t^2  = 2\bv{A}_t$}\\
&\le \tr \left ( \E [\bv{B}_{t-1} \bv{B}_{t-1}^T] \right  ) \cdot \norm{\bv{I}+2(\eta+\eta^2)(\bv{D}+\bv{W})}_2\\
&\le \alpha_{t-1} \cdot \left (1+ 4 \eta \lambda_1(\bv{D}+\bv{W})\right)
\end{align*}
where the second to last bound follows since $\bv{D}+\bv{W}$ and $\bv{B}_t\bv{B}_t$ are PSD. The last bound follows since $\eta < 1$. We thus have, applying this argument inductively, $\alpha_T \le \left (1+ 4 \eta \lambda_1(\bv{D}+\bv{W})\right)^T$, and so, $\norm{\bv{B}_t}_2 \le \norm{\bv{B}_t}_2^2 \le \tr(\bv{B}_t \bv{B}_t^T) \le \frac{\left (1+ 4 \eta \lambda_1(\bv{D}+\bv{W})\right)^T}{\delta/3}$ with probability $\ge 1-\delta/3$ by Markov's inequality. Combined with our bound on $\norm{\bv{Q}_0}_2$ this gives, with probability $\ge 1-\delta/3-\delta/3$:

\begin{align*}
\smax(\bv{ Q}) = \norm{\bv{ Q}}_2 &\le \prod_{j=1}^T \norm{\bv{I} + \eta \bv{A}_j}_2 \cdot \norm{\bv{Q}_0}_2\\
&\le \frac{(1+4\eta \lambda_1(\bv{D}+\bv{W}))^{\frac{c_2(\log \xi + 1/\epsilon)}{\bgap \cdot \eta}}}{\delta}  \cdot \norm{\bv{Q}_0}_2\\
& \le e^{\frac{c(\log \xi + 1/\epsilon)\cdot  \lambda_1(\bv{D}+\bv{W})}{\bgap}} \cdot \frac{\norm{\bv{Q}_0}_2}{\delta/3}\\
&\le e^{\frac{c(\log \xi + 1/\epsilon)\cdot \lambda_1(\bv{D}+\bv{W})}{\bgap}+c\log \xi}\\ &\le e^{c(\log \xi + 1/\epsilon)\cdot \left [\frac{ \lambda_1(\bv{D}+\bv{W})}{\bgap}+1\right]}
\end{align*}
for some large enough $c$. In the second to last step we bound $\frac{\norm{\bv{Q}_0}_2}{\delta/3} \le \frac{cnk \sqrt{\log(1/\delta)}}{\delta/3} \le c\log \xi$ for large enough $c$.
The theorem follows after union bounding, which gives that both our upper and lower bounds hold with probability  $\ge 1-\delta$. 
\end{proof}

We can finally complete our analysis, proving our main asynchronous eigenvector approximation theorem, restated below:

\begin{reptheorem}{cor:ojaAsync2}
Let $\bv{v}_1,...,\bv{v}_k$ be the top $k$ eigenvectors of the communication matrix $\bv{D}+\bv{W}$ in an asynchronous communication model, and let $\Lambda,\bgap, \gamma_{mix}$ be  bounds satisfying: $\Lambda \ge \sum_{j=1}^k \lambda_j(\bv{D}+\bv{W})$, $\bgap \le \min_{j \in [k]} [\lambda_j(\bv{D}+\bv{W})-\lambda_{j+1}(\bv{D}+\bv{W})]$, and $\gamma_{mix} \le \min \left [ \frac{1}{n},\log \left(\lambda_2^{-1}(\bv{I}-\frac{1}{2}\bv{D} + \frac{1}{2}\bv{W})\right)\right ]$.

For any $\epsilon,\delta \in(0,1)$, let $\xi = \frac{n}{\delta \epsilon \cdot \bgap}$. Let $\eta =\frac{c_1 \epsilon^2 \cdot \bgap \cdot \delta^2}{\Lambda k^3 \log^3 \xi}$ for sufficiently small $c_1$, and $T = \frac{c_2 \cdot (\log \xi + 1/\epsilon)}{\bgap \cdot \eta}$, $T' = \frac{c_3(\log \xi + 1/\epsilon)\cdot \lambda_1(\bv{D}+\bv{W})}{\bgap \cdot \gamma_{mix}}$ for sufficiently large $c_2,c_3$.
 For all $u \in [n],i \in [k]$, let ${\hat v}_u^{(j)}$ be the local state computed by \ripref{alg:ojaGossip}.
If $\bv{\hat V} \in \R^{n \times k}$ is given by $(\bv{\hat V})_{u,j} = {\hat v}_u^{(j)}$
 and $\bv{\hat v}_i$ is the $i^{th}$ column of $\bv{\hat V}$, then with probability $\ge 1 -\delta-e^{-\Theta(n)}$, for all $i \in [k]$:
 \begin{align*}
 \left |\bv{\hat v}_i^T \bv{v}_{i}\right| \ge 1-\epsilon\hspace{1em}\text{ and }\hspace{1em}\norm{\bv{\hat v}_i}_2 \le 1+\epsilon.
 \end{align*}
\end{reptheorem}
\begin{proof}
Let $\bv{ Q} \in \mathbb{R}^{n \times k}$ be given by $(\bv{ Q})_{u,j} = {q}_u^{(j)}$, after round $T$ (i.e. $\bv{Q}$ is the input passed to \texttt{AsynchOrth} by \texttt{AsynchOja}.
Let $\bv{\tilde V}$ be the output given by running exact Cholesky orthogonalization (\ripref{alg:chol}) on $\bv{Q}$. Let $\bv{\tilde v}_i$ be its $i^{th}$ column.
Let $\epsilon' =  \frac{\epsilon \delta^3}{n^2 \cdot exp\left( \frac{c(\log \xi + 1/\epsilon)\cdot \lambda_1(\bv{D}+\bv{W})}{\bgap}\right)}$. 
By our bound on $\sigma_{\max}(\bv{Q}) = \norm{\bv{Q}}_2$ in \autoref{lem:qcond}, with probability  $1-\delta-\epsilon^{-\Theta(n)}$, $T'$  satisfies $
T' = \Omega \left ( \frac{\log \left ( \frac{n \norm{\bv{Q}}_2}{\epsilon'  \delta} \right  )}{\gamma_{mix}} \right ).$
So by \autoref{lem:dotProducts}, conditioned on the previous bound holding, with probability  $\ge 1-\delta$, for all $u$ and $i,j\in[k]$, $\bv{R}_u$ computed in \ripref{alg:pporth} satisfies: $$\left |(\bv{R}_u)_{i,j} - (\bv{ Q}^T \bv{ Q})_{i,j} \right | \le \epsilon'.$$
We can then
apply \autoref{lem:stabilityAnalysis} and our bounds on the maximum and minimum singular values of $\bv{Q}$ in  \autoref{lem:qcond}. With probability  $1-2\delta-e^{-\Theta(n)}$, for all $i \in [k]$:
\begin{align}\label{normBound}
\norm{\bv{\hat v}_i  -  \bv{\tilde v}_i}_2 \le \frac{2\sqrt{2} \epsilon' n^2 e^{ \frac{c(\log \xi + 1/\epsilon) \cdot\lambda_1(\bv{D}+\bv{W})}{\bgap}}}{(\delta/c)^3} = O(\epsilon).
\end{align}

Now, by \autoref{cor:ojaAsync} and our setting of $\eta$ and $T$,
with probability  $1-\delta$, for all $i \in[k]$, $\left |\bv{\tilde  v}_i^T\bv{v}_i \right | \ge 1-\epsilon$. 
Since $\norm{\bv{v}_i}_2 = 1$, \eqref{normBound} gives, with overall probability $\ge 1-3\delta - e^{-\Theta(n)}$:
\begin{align*}
 \left |\bv{\hat v}_i^T \bv{v}_i \right | = \left |\bv{\tilde  v}_i^T \bv{v}_i - (\bv{\tilde v}_i^T  - \bv{\hat v}_i^T) \bv{v}_i \right |&\ge \left |\bv{\tilde  v}_i^T \bv{v}_i\right  | - \left |(\bv{\tilde v}_i^T  - \bv{\hat v}_i^T) \bv{v}_i \right |\\
&\ge 1-\epsilon - \norm{\bv{\hat v}_i - \bv{\tilde v}_i}_2\\
&\ge 1-O(\epsilon).
\end{align*}
Additionally, using \eqref{normBound} and $\norm{\bv{\tilde v}_i}_2 = 1$ (shown in \autoref{cor:ojaAsync}), we can apply triangle inequality to show $\norm{\bv{\hat v}_i}_2 \le 1 + O(\epsilon)$. The theorem follows by adjusting constants on $\epsilon,\delta$.
\end{proof}

\section{Community Detection in the \modelpq}\label{app:gnp}

To convert our matrix concentration bound to a bound on the closeness of the two eigenvectors, we apply the
 Davis-Kahan theorem: 

\begin{theorem}[Davis-Kahan Theorem -- \cite{demmel1997applied} Theorem 5.4]\label{davisKahan} Let $\bv{M}$ and $\bv{\hat M}$ be symmetric matrices with eigenvectors $\bv{v}_1,...,\bv{v}_n$ and $\bv{\hat v}_1,...,\bv{\hat v}_n$ respectively. We have:
\begin{align*}
1 - \left |\bv{\hat  v}_i^T \bv{v}_i\right | \le \frac{2\norm{\bv{M}-\bv{\hat M}}_2}{\min\left [\lambda_{i-1}(\bv{M})-\lambda_{i}(\bv{M}),\lambda_{i}(\bv{M})-\lambda_{i+1}(\bv{M})\right]}.
\end{align*}
\end{theorem}
\begin{proof}
Theorem 5.4 of \cite{demmel1997applied} states the above bound with the lefthand side equal to $\sin 2\theta_i$, where $\theta_i$ is the angle between $\bv{v}_i$ and $\bv{\hat  v}_i$. It is noted in the proof that the bound can also be proven  on $\sin \theta_i = \sqrt{1-\cos^2 \theta_i} = \sqrt{1-(\bv{\hat  v}_i^T \bv{v}_i)^2} \ge 1 - \left |\bv{\hat  v}_i^T \bv{v}_i\right |$, giving our statement of the bound.
\end{proof}
Using this theorem we show:
\begin{lemma}[Concentration of $G(n,p,q)$ communication matrix second eigenvector]\label{gnpConcentration}
Let  $\bv{v}_2$ be the second eigenvector of $\bv{C}^{(p,q)}$ for any $p > q > 0$. 
Let $\bv{\tilde v}_2$ be the second eigenvector of $(\bv{D}+\bv{W})$, where $\bv{W}$ is the communication weight matrix and  $\bv{D}$ is the degree matrix in the \modelg. Then if $\frac{\min\left  [q, p-q \right]}{\sqrt{p+q}} \ge \frac{9\sqrt{\log(n/\delta)}}{\epsilon \sqrt{n}}$, with probability  $\ge 1-\delta$:
\begin{align*}
\left | \bv{\tilde v}_2^T  \bv{v}_2 \right | \ge  1-\epsilon.
\end{align*}
\end{lemma}
\begin{proof}
Consider the $n\times {n \choose 2}$ matrix $\bv{B}$ with columns indexed by unordered pairs of vertices, $(u,v)$ with $u \neq v$. If $u$ and $v$ are in the same cluster, the $(u,v)$ column is $\sqrt{p} \cdot \euv$. If they are in different clusters it is $\sqrt{q}\cdot \euv$. We can see  that $\bv{B} \bv{B}^T = \bv{C}^{(p,q)} + \frac{n(p+q)-2p}{2} \cdot \bv{I}_{n \times n}$. Thus, the second eigenvector of $\bv{B}\bv{B}^T$ is identical to that of $\bv{C}^{(p,q)}$, $\bv{v}_2$. Further we can compute the eigenvalues:
\begin{align*}
\lambda_1(\bv{B}\bv{B}^T) &= \lambda_1(\bv{C}^{(p,q)}) + \frac{n(p+q)-2p}{2}  = n(p+q)-p\\
\lambda_2(\bv{B}\bv{B}^T) &= \lambda_2(\bv{C}^{(p,q)}) + \frac{n(p+q)-2p}{2}  = np-p\\
\lambda_i(\bv{B}\bv{B}^T) &= \frac{n(p+q)}{2}-p\text{ for all $i \in \{3,...,n\} $}.
\end{align*}

Let $\bv{\tilde B} \bv{\tilde B}^T$ be obtained from $\bv{B}$ by independently setting each column to $\bv{b}_i/\sqrt{p_i}$ with probability $p_i$ equal to its squared norm ($p$ for intracluster columns and $q$ for intercluster columns) and to $\bv{0}$ otherwise. In this way $\E[\bv{\tilde B} \bv{\tilde B}^T] = \bv{B}\bv{B}^T$. Further, we can see that sampling $\bv{\tilde B} \bv{\tilde B}^T$ is identical to sampling $\bv{D}+\bv{W}$ in the \modelg, up to a scaling. Thus, to prove the lemma, it suffices to show that $\bv{\tilde v}_2$, the second eigenvector of $\bv{\tilde B} \bv{\tilde B}^T$ is close to $\bv{v}_2$ with good probability. We do this via a matrix Bernstein bound, which shows that 
$\bv{\tilde B} \bv{\tilde B}^T$ is close to $\bv{ B} \bv{ B}^T$ with good probability.

Specifically let $\bv{b}_i$ and $\bv{\tilde b}_i$ be the $i^{th}$ columns of $\bv{B}$  and $\bv{\tilde B}$  respectively.
We can apply Theorem 1.4 of \cite{tropp2012user}, where in the notation of the theorem we set $\bv{X}_k = \bv{\tilde b}_k\bv{\tilde b}_k^T - \bv{b}_k\bv{ b}_k^T$. We  have $\E [\bv{X}_k] = \bv{0}$ and $\norm{\bv{X}_k} \le 2$ always, and
\begin{align*}
\sigma^2  = \left  \| \sum_k \E(\bv{X}_k^2)\right  \|_2 &= \left  \|\sum_{k} \left (\E[\norm{\bv{\tilde b}_k}_2^2 \cdot \bv{\tilde b}_k \bv{\tilde b}_k^T]-2  \E [\bv{\tilde b}_k \bv{\tilde b}_k^T \bv{b}_k \bv{b}_k^T] + \norm{\bv{b}_k}_2^2\bv{b}_k \bv{b}_k^T \right  )\right \|_2\\
&= \left  \|\sum_{k} \left (2 \bv{b}_k \bv{b}_k^T - \norm{\bv{b}_k}_2^2\cdot \bv{b}_k \bv{b}_k^T\right) \right \|_2\tag{Since $\norm{\bv{\tilde b}_k}_2^2 = 2$ with probability  $\norm{\bv{b}_k}_2^2$, $0$ otherwise.}\\
&\le \norm{  2 \bv{B}\bv{B}^T}_2 \le  2n(p+q).
\end{align*}
where the final bound follows by our computation of $\lambda_1(\bv{B}\bv{B}^T) = \norm{  \bv{B}\bv{B}^T}_2$ above.
The inequality  follows from the fact that $2 - \norm{\bv{b}_k}_2^2 \in [2-p^2,2-q^2]$ and so $\sum_{k} \left (2 \bv{b}_k \bv{b}_k^T - \norm{\bv{b}_k}_2^2\cdot \bv{b}_k \bv{b}_k^T\right)$ is PSD and $\preceq (2-q^2)\bv{B}\bv{B}^T \preceq 2 \bv{B}\bv{B}^T$.

Plugging the above bounds into Theorem 1.4 of \cite{tropp2012user} gives:
\begin{align}\label{troppBound}
\Pr  \left  [\norm{\bv{\tilde B} \bv{\tilde B}^T - \bv{B} \bv{B}^T}_2 > \Delta \right  ] \le n \cdot \exp \left (\frac{-\Delta^2/2}{2n(p+q)+2\Delta/3} \right  ).
\end{align}

Applying the Davis-Kahan theorem (\autoref{davisKahan}) to $\bv{BB}^T$  and $\bv{\tilde B  \tilde B}^T$, we have using our eigenvalue calculations for $\bv{BB}^T$:
\begin{align}\label{BBgaps}
\min\left [\lambda_{1}(\bv{BB}^T)-\lambda_{2}(\bv{BB}^T),\lambda_{2}(\bv{BB}^T)-\lambda_{3}(\bv{BB}^T)\right] = \min \left  (q\cdot n, \frac{n(p-q)}{2} \right).
\end{align}
and so $1 - \left |\bv{\tilde  v}_2^T \bv{v}_2\right | \le \epsilon$ as long as $\Delta \le \epsilon \cdot  \min \left  [\frac{q\cdot n}{2}, \frac{n(p-q)}{4} \right].$ Plugging this in  to \eqref{troppBound} and simplifying:
\begin{align}\label{finalBernstein}
\Pr  \left  [\norm{\bv{\tilde B} \bv{\tilde B}^T - \bv{B} \bv{B}^T}_2 > \epsilon \cdot  \min \left  [\frac{q\cdot n}{2}, \frac{n(p-q)}{4} \right] \right  ] &\le n \cdot \exp \left (\frac{-\epsilon^2 \cdot \min\left  [\frac{q^2\cdot n^2}{4}, \frac{n^2(p-q)^2}{16} \right]/2}{2.166 \cdot n (p+q)} \right  )\nonumber\\
&\le n \cdot \exp \left (\frac{-\epsilon^2 n \cdot \min\left  [q^2, (p-q)^2 \right]}{70(p+q)} \right  )
\end{align}
This probability is bounded by $\delta$ if $\frac{\min\left  [q, p-q \right]}{\sqrt{p+q}} \ge \frac{\sqrt{70 \log(n/\delta)}}{\epsilon \sqrt{n}}$, giving the lemma.

\end{proof}

With \autoref{gnpConcentration} ensuring that the second eigenvector of $\bv{D}+\bv{W}$ in fact  approximates the cluster indicator vector $\bs{\chi}$, we can now show that  approximately computing this eigenvector using the \texttt{AsynchOja} algorithm and thresholding its entries by their signs gives an approximately  correctly distributed community detection protocol.

\begin{reptheorem}{thm:modelg}
Consider a set of nodes executing \texttt{AsynchCD}$(T,T',\eta)$ (\ripref{alg:commDetec}) in the \modelg. Let $\rho = \min \left (\frac{q}{p+q}, \frac{p-q}{p+q} \right )$. If for sufficiently small constant $c_1$ and sufficiently large constants $c_2,c_3$,
\begin{align*}
\hspace{-.5em}\eta =\frac{c_1 \epsilon^2 \delta^2 \cdot \rho}{\log^3 \left (\frac{n}{\epsilon \delta \rho} \right )},\hspace{.5em}\hspace{.5em}T = \frac{c_2 \cdot n \cdot \left (\log^3 \left (\frac{n}{\epsilon \delta \rho} \right ) + \frac{\log \left (\frac{n}{\epsilon \delta \rho} \right )}{\epsilon}\right)}{\epsilon^2 \delta^2 \rho^2},\hspace{.5em}\text{ and }\hspace{.2em}T' = \frac{c_3\cdot n \cdot \left (\log \left (\frac{n}{\epsilon \delta \rho} \right ) + \frac{1}{\epsilon} \right )}{\rho^2},
\end{align*}
and if $\frac{\min\left  [q, p-q \right]}{\sqrt{p+q}} \ge \frac{c_4\sqrt{\log (n/\delta)}}{\epsilon \sqrt{n}}$ for large enough constant $c_4$,
then, with probability  $1-\delta$, after ignoring $\epsilon \cdot n$ nodes, all remaining nodes in $V_1$ terminate in some state $s_1$, and all remaining nodes in $V_2$ terminate in some different state $s_2$. Supressing polylogarithmic factors in the parameters, the total number of global rounds and local rounds required are: $T+T' = \tilde O \left ( \frac{n}{\epsilon^3 \delta^2 \rho^2} \right )$ and $L = \tilde O \left (\frac{1}{\epsilon^3 \delta^3  \rho^2} \right)$.
\end{reptheorem}
\begin{proof}
Let  $\bv{v}_2 = \bs{\chi}$ be the second eigenvector of $\bv{C}^{(p,q)}$ and $\bv{\tilde v}_2$ be the second eigenvector of $\bv{D}+\bv{W}$. By \autoref{gnpConcentration}, with probability $1-\delta$, $|\bv{\tilde v}_2^T \bv{v}_2| \ge  1-\epsilon$ which gives if $\bv{\tilde v}_2^T \bv{v}_2 \ge 0$,
\begin{align*}
\epsilon \ge 1 - \bv{\tilde v}_2^T \bv{v}_2 = \bv{\tilde v}_2^T(\bv{\tilde v}_2 - \bv{v}_2)
\end{align*}
and so $\norm{\bv{\tilde v}_2 - \bv{v}_2}_2 \le \epsilon$. Similarly, if $\bv{\tilde v}_2^T \bv{v}_2 \le 0$ we have $\norm{\bv{\tilde v}_2 + \bv{v}_2}_2 \le \epsilon$.
If $\bv{\hat v}_2$ satisfies $|\bv{\hat  v}_2^T \bv{\tilde v}_2| \ge 1-\epsilon$ and $\norm{\bv{\hat v}}_2 \le 1+\epsilon$ then this gives:
\begin{align*}
|\bv{\hat  v}_2^T \bv{v}_2| \ge |\bv{\hat  v}_2^T \bv{\tilde v}_2| - |\bv{\hat  v}_2^T (\bv{v}_2 - \bv{\tilde v}_2)| \ge 1-(2+\epsilon) \epsilon 
\end{align*} 
in the case where $\bv{\tilde  v}_2^T \bv{ v}_2 \ge 0$ and
\begin{align*}
|\bv{\hat  v}_2^T \bv{v}_2| \ge |-\bv{\hat  v}_2^T \bv{\tilde v}_2| - |\bv{\hat  v}_2^T (\bv{v}_2 + \bv{\tilde v}_2)| \ge 1-(2+\epsilon)\epsilon
\end{align*} 
in the case where $\bv{\tilde  v}_2^T \bv{ v}_2 \le 0$. Either way, we have:
\begin{align}\label{eigenConversion}
|\bv{\hat  v}_2^T \bv{v}_2| = 1-O(\epsilon).
\end{align}

Now, in the \modelg\ $\bv{D} + \bv{W}$ is a scaling of $\bv{\tilde B} \bv{\tilde B}^T$, as defined in the proof of \autoref{gnpConcentration}. The scale factor is the inverse of the number of edges in the sampled graph which is $\frac{n^2(p+q)-2np}{4}$ in expectation. As long as $p+q  \ge \frac{c\log(n/\delta)}{n}$ for large enough $c$ (which is implied by  our assumption $\frac{\min\left  [q, p-q \right]}{\sqrt{p+q}} \ge \frac{c_4\sqrt{\log (n/\delta)}}{\epsilon \sqrt{n}}$), then by a simple Chernoff bound, the number of sampled edges will be within a $2$ factor of this expected value with probability  $1-\delta$. Thus:
\begin{align}\label{eq:scale}
\bv{D} + \bv{W} =  s \cdot  \bv{\tilde B} \bv{\tilde B}^T \hspace{1em} \text { for some }\hspace{1em} s \ge \frac{2}{n^2(p+q)}.
\end{align}
Again applying the matrix Bernstein bound in \autoref{gnpConcentration}, equation \eqref{finalBernstein}, with probability  $\ge 1-\delta$:
\begin{align}\label{camPertrub}
\norm{\bv{\tilde B}\bv{\tilde B}^T-\bv{ B}\bv{ B}^T}_2 \le \epsilon \min \left  [\frac{q\cdot  n}{2},\frac{n(p-q)}{4} \right  ].
\end{align}

Combined with the eigenvalue gap calculations for $\bv{BB}^T$ shown in \eqref{BBgaps}, we have:
\begin{align}\label{Btildegaps}
\min\left [\lambda_{1}(\bv{\tilde B \tilde B}^T)-\lambda_{2}(\bv{\tilde B \tilde B}^T),\lambda_{2}(\bv{
\tilde B \tilde B}^T)-\lambda_{3}(\bv{\tilde B \tilde B}^T)\right] \ge (1-\epsilon) \cdot \min \left  (q\cdot n, \frac{n(p-q)}{2} \right).
\end{align}
After scaling, by \eqref{eq:scale}:
\begin{align*}
\hspace{-1em}\min\left [\lambda_{1}(\bv{D}+\bv{W})-\lambda_{2}(\bv{D}+\bv{W}),\lambda_{2}(\bv{D}+\bv{W})-\lambda_{3}(\bv{D}+\bv{W})\right] &\ge \frac{2(1-\epsilon)}{n^2(p+q)} \cdot \min \left  (q n, \frac{n(p-q)}{2} \right)\\
&\ge \frac{\rho}{2n}
\end{align*}
where in the final bound we assume $1-\epsilon \ge 1/2$ which is without loss of generality, since we can always scale $\epsilon$ down by  a $1/2$ factor.

We similarly use the perturbation bound  of \eqref{camPertrub}, the scale bound of \eqref{eq:scale} and our eigenvalue calculations for $\bv{BB}^T$ to argue that that $\lambda_1(\bv{D}+\bv{W}) \le \frac{16}{n}$ and that $\lambda_2(\bv{I} - 1/2\bv{D} + 1/2\bv{W}) \le 1-\frac{q}{2n(p+q)}$ and so $\log(\lambda_2^{-1}(\bv{I} - 1/2\bv{D} + 1/2\bv{W}) \ge \frac{q}{2n(p+q)} \ge \frac{\rho}{2 n}$.		
													
We can thus apply \autoref{cor:ojaAsync2} with $k = 2$, $\Lambda = \frac{32}{n}$, $\bgap = \frac{\rho}{2 n}$, and $\gamma_{mix} = \frac{\rho}{2 n}$. With these parameters
 we can set,  for sufficiently small $c_1$ and large $c_2,c_3$,
\begin{align*}
\hspace{-.5em}\eta =\frac{c_1 \epsilon^2 \delta^2 \cdot \rho}{\log^3 \left (\frac{n}{\epsilon \delta \rho} \right )},\hspace{.5em}T = \frac{c_2 \cdot n \cdot \left (\log^3 \left (\frac{n}{\epsilon \delta \rho} \right ) + \frac{\log \left (\frac{n}{\epsilon \delta \rho} \right )}{\epsilon}\right)}{\epsilon^2 \delta^2 \rho^2},\hspace{.5em}\text{ and }\hspace{.2em}T' = \frac{c_3\cdot n \cdot \left (\log \left (\frac{n}{\epsilon \delta \rho} \right ) + \frac{1}{\epsilon} \right )}{\rho^2}.
\end{align*}
where to bound $T'$ we use that $\frac{\lambda_1(\bv{D}+\bv{W})}{\bgap} \le \frac{32}{\rho}$.
Letting $\bv{\hat V} \in \R^{n \times k}$ be given by $(\bv{\hat V})_{u,j} = {\hat v}_u^{(j)}$ where ${\hat v}_u^{(j)}$ are the outputs of $\texttt{AsynchOja}(T,T',\eta)$ and letting $\bv{\hat v}_2$ be the second column of $\bv{\hat V}$, with these parameters, \autoref{cor:ojaAsync2} gives  that with probability $1-\delta$,
 \begin{align*}
 \left |\bv{\hat v}_2^T \bv{\tilde  v}_{2}\right| \ge 1-\epsilon\hspace{1em}\text{ and }\hspace{1em}\norm{\bv{\hat v}_2}_2 \le 1+\epsilon
 \end{align*}
 where $\bv{\tilde  v}_2$ is the second eigenvector of $\bv{D}+\bv{W}$. By \eqref{eigenConversion} we thus have $\left  |\bv{\hat v}_2^T \bv{v}_2 \right | \ge 1-O(\epsilon)$.
 
 Applying \autoref{approx_eig_starting_point} after adjusting $\epsilon$ by  a constant  factor then gives the theorem. Additionally, in expectation, each node is involved in $L = \Theta \left ( \frac{T+T'}{n}\right  )$ interactions. This bound also holds for all nodes with probability $1-\delta$ by a Chernoff bound, since $L = \Omega(\log(n/\delta))$. We can union bound over the various events required for the theorem to hold, all occurring with probability $\ge 1-\delta$, which gives the final theorem after adjusting $\delta$ by a constant factor. 
 \end{proof}

\section{Analysis of  \algoname{Cleanup Phase} in \modelpq}\label{sec:cleanpq}

As we will argue in \autoref{sec:cleang}, if $q \leq p/2$, then the global number of rounds simply becomes $O(n \log^2 n)$.

\begin{lemma}\label{lem:onephasepq}
	Consider \modelpq.  After $O\left(\frac{72 n \ln n  }{ (\sqrt{p'}-\sqrt{q'})^2  }\right)$ global rounds Algorithm \algoname{Cleanup Phase} with parameters $k=1$ 
	and $r=\frac{72 n \ln n  }{ (\sqrt{p'}-\sqrt{q'})^2  }$, where $p'=\frac{(1-\epsilon)p}{p+q}>0$ and $q'=\frac{q+\epsilon p}{p+q}$, $p'>q'$
	we have that all nodes are correctly labeled.
\end{lemma}
\begin{proof}
	 Fix an arbitrary node $u$. 
Let $X_u$ ($X'_u)$, respectively) denote the number of times an edge $(u,v)$ was chosen where $v$'s current label is the same as $u$'s ground-truth label meaning that both are in $V_1$ or $V_2$ (opposite, respectively). Such an edge is scheduled at a given time step with probability at least $p_1 \eqdef \frac{(1-\epsilon)n/2p}{p\left(\frac{n}{2}\right)^2
+q\left(\frac{n}{2}\right)^2
} =\frac{2(1-\epsilon)p}{(p+q)n}$ (at most $p_2 \eqdef \frac{2(q+\epsilon p)}{(p+q)n}$, respectively).

We show that w.h.p. $X_u > X'_u$; taking Union bound over the complementary events for all nodes  yields the claim.

We distinguish between two cases. First assume $q' \geq p'/12$.
Let $\delta=\sqrt{6 \ln n/\Ex{X_u}} , \delta'=\sqrt{6 \ln n/\Ex{X'_u}}$ and observe that $\delta,\delta'\leq 1$.
By Chernoff bounds with, we get,

\begin{align}
	\Pro{X_u \geq (1-\delta)\Ex{X_u} } + \Pro{X'_u \leq (1+\delta')\Ex{X'_u} }   
	\leq 2 e^{-2 \ln n}.
\end{align}

Conditioning on this, we have 
\begin{align}
X_u -X_u'\geq (1-\delta)\Ex{X_u}- (1+\delta)\Ex{X'_u}  = r p_1-\sqrt{6 \ln nrp_1} - 
\left( r p_2 +\sqrt{6\ln n rp_2} \right) > 0,
\end{align}
for $r=\frac{72 n \ln n  }{ (\sqrt{p'}-\sqrt{q'})^2  } \geq \frac{6\ln n}{(\sqrt{p_1}-\sqrt{p_2})^2}.$
Similarly, if $q' \leq p'/12$, then, by \autoref{chernoff},
\begin{align}
	\Pro{X_u < X'_u} &\leq \Pro{X_u < \Ex{X_u}/2 } + \Pro{X'_u \geq 6\max\{ \ln n,\Ex{X'_u} \}} \leq   
	 2 e^{-2 \ln n}.
\end{align}

Thus, in both cases we get
\begin{align}
	\Pro{X_u \leq X'_u} &	\leq 2 e^{-2 \ln n}.
\end{align}
This finishes the proof.
\end{proof}

\begin{proof}[Proof of \autoref{thm:cleanuppq}]
The lemma is a direct consequence of 
\autoref{lem:onephasepq}.
\end{proof}

\section{Analysis of  \algoname{Cleanup Phase} in \modelg}\label{sec:cleang}

Let $E(u,S)$ denote the number of edges between $u\in V$ and $S \subseteq V$.
Recall that \[\Delta = \frac{p}{2}-\frac{q}{2}-\sqrt{12p\ln n/n}-\sqrt{12q\ln n/n}. \] 
Throughout this section we will assume that 
\begin{align}\label{eq:assumptions}
 \Delta \geq 2^{15} \ln n/n. 
\end{align}

 Let $n^* = \frac{n\Delta}{240p}$.
Fix an arbitrary set $S \subseteq V$ with $|S| \leq  n^*.$
For $u\in V$ define \[ Y_u^S 
=\begin{cases}
	1 & \text{ $E(u,S) \geq \Delta n /12$} \\
	0 & \text{ otherwise}
\end{cases}.
\]
%
%
%
%
Let
 \[ Y^S 
=\begin{cases}
	1 & \text{ $\sum_{u \in V} Y_u^S \geq |S|/3$} \\
	0 & \text{ otherwise}
\end{cases}.
\]
\begin{lemma}
Assume that $|S| \leq n^*$.
We have that
	\[ \Pro{ Y_S	 =1 } \leq e^{- \Delta |S| n /2^{11}} .\]
\end{lemma}
\begin{proof}
Observe that the $\{ Y_u^S \colon u\not \in S\}$ are independent $0/1$-variables. 
We remark that $\{ Y_u^S \colon u\in S\}$, are correlated and we will need a different approach to bound $\sum_{u\in S} Y_u^S$ which we save for later.
Applying Chernoff with $\delta = \frac{n\Delta}{p |S|12} -1$, $\Ex{E(u,S)}=p|S|$. Note that $\frac{n\Delta}{p |S|12} -1 \geq 
\frac{n\Delta}{p |S|24} \geq 1 $, due to the bound on $|S|$. We get
\begin{align}\label{eq:ranoutofnames} p' &\eqdef  \Pro{ Y^S_u	 =1 } = \Pro{E(u,S) \geq n \Delta/12} =\Pro{E(u,S) \geq (1+\delta)\Ex{E(u,S)}  } \\
&\leq \exp\left(- \frac{p|S|}{3}\left(\frac{n\Delta}{p |S|12} -1\right)  \right)
\leq e^{-\Delta n/72 } . \end{align}

where the penultimate inequality stems from Chernoff bounds.
	By \autoref{thm:bounds-binomial-distribution}, with parameters $\alpha = \frac{|S|}{6m}, m=n-|S|, p=p'$,
	we get that
	\begin{align*}
	 \Pro{\sum_{u \in V\setminus S} Y_u^S \geq |S|/6}  &\leq 
	\left(\left(\frac{p'}{\frac{|S|}{6m}} \right)^{\frac{|S|}{6m}} \left(\frac{1-p'}{1-\frac{|S|}{6m}} \right)^{1-\frac{|S|}{6m}}\right)^{m}\\
	&\leq 
	\left(\left(\frac{p'}{\frac{|S|}{6m}} \right)^{\frac{|S|}{6m}} \left(\frac{1}{1-\frac{|S|}{6m}} \right)\right)^m 
	\leq 
	\left(\left( 6mp' \right)^{\frac{|S|}{6m}} \right)^m  \left(\frac{1}{1-\frac{|S|}{6m}} \right)^m
\\&	\leq \left( 6m \right)^{\frac{|S|}{6}} \left( p' \right)^{\frac{|S|}{6}} \left(\frac{1}{1-\frac{|S|}{6n}} \right)^m
\\
		&\leq 
		\exp \left(  \frac{|S|}{6}\ln(6m) - \frac{|S|}{6} \Delta n/72 +m \ln \left( \frac{1}{1-\frac{|S|}{6m}} \right)\right)\\
				&\leq 
		\exp \left(  \frac{|S|}{6}\ln(6m) - \frac{|S|}{6} \Delta n/72 +m \frac{e|S|}{6m}\right)\\
		&
		\leq 
		\exp \left( -  \frac{|S| \Delta n}{2^{10}} \right),
	\end{align*}
	where the penultimate inequality comes from $\ln(1/(1-x))\leq ex$ for  $x \in (0,0.8]$ and the last inequality comes from \eqref{eq:assumptions}.

	We now turn to bounding $\{ Y_u^S \colon u\in S\}$. 
	In order for $\sum_{u \in S} Y_u^S \geq |S|/6 $ a counting argument shows that the number of required edges with both endpoints in $S$ needs to be at least $\frac{\Delta n}{12}\frac{|S|}{6} \frac{1}{2}$.
	we have $\Ex{E(S,S)}=p|S|^2/2$. 
	In
	Applying Chernoff bounds yields $\delta=\frac{n\Delta}{72p|S|} -1 \geq  \frac{n\Delta}{144p|S|} \geq 1$
	\begin{align*}
		\Pro{\sum_{u \in S} Y_u^S \geq |S|/6}   &\leq \Pro{ E(S,S) \geq \frac{\Delta n}{12}\frac{|S|}{6}} = \Pro{ E(S,S) \geq (1+\delta)\Ex{E(S,S)}}
		\\&\leq 
		\exp \left( -  \frac{|S| \Delta n}{2^{10}} \right).
	\end{align*}

	We have
	\begin{align*}
			\Pro{Y^S =1} &\leq \Pro{\sum_{u \in V\setminus S} Y_u^S \geq |S|/6}  + \Pro{\sum_{u \in S} Y_u^S \geq |S|/6}  \leq 
		2\exp \left( -  \frac{|S| \Delta n}{2^{10}} \right)
	\end{align*}
	
\end{proof}

We say a graph is \emph{smooth} if for all subsets $S$ of size $|S| \in [ \nzero,n^*]$ we have
$Y^S=0$.
\begin{lemma}\label{lem:smooth}
	Let $G_{n,p}$ be an Erd\H{o}s-R\'eny graph with parameters $n,p$ satisfying \eqref{eq:assumptions}.
Then,
  $G_{n,p}$ is smooth w.h.p..
  
\end{lemma}
\begin{proof}
%

Applying Union bound,
\begin{align}
\Pro{ \sum_{S, |S|\in [\nzero, n^*]} Y_S	 =0 }
&\leq \sum_{S, |S|\in [\nzero ,n^*]}  \Pro{ Y_S	 =0 }
\leq \sum_{i=\nzero}^{n^*} \sum_{S, |S|=i}  \Pro{ Y_S	 =0}\\
&\leq \sum_{i=\nzero}^{n^*}  {n \choose i}\Pro{ Y_S	 =0~|~|S|=i  }
\leq \sum_{i=\nzero}^{n^*}  \left(\frac{e n}{ i} \right)^i  e^{- \Delta n i/2^{11}} \\
&\leq \sum_{i=\nzero}^{n^*}   \exp\left( i \ln n -\Delta n i/2^{11}\right)
 \leq \sum_{i=\nzero}^{n^*} \frac{1}{n^3}
\leq \frac{1}{n^2}.
\end{align}

\end{proof}

\begin{lemma}\label{lem:exepcteddrop}
	Consider \modelg. Assume  the graph is smooth. 
	Let $S_t$ be the set of nodes that are incorrectly labeled after phase $t$.
	If $|S_t|\in [ \nzero,n^*]$, then after one additional phase of $O\left(\frac{ n p\log n  }{ (\sqrt{p''}-\sqrt{q''})^2  }\right)$ global rounds Algorithm \algoname{Cleanup Phase} with parameter	 $r=\frac{72 p n \log n  }{ (\sqrt{p''}-\sqrt{q''})^2  }$, where
	\begin{enumerate}
		\item $p''= \frac{p}{2}-\sqrt{\frac{6 p\log n}{n}}  -\frac{\Delta}{12}$ and
		\item  $q''=\frac{q}{2} + \sqrt{\frac{6 q\log n}{n}}  +\frac{\Delta}{12}$ and 
		\item 	 $\Delta=\Omega(\log n/n)$,
	\end{enumerate}
	conditioning  on $ \mathcal{F}_t$ we get that  w.h.p. $|S_{t+1}|  \leq (2/3) |S_t|$,
	where $\mathcal{F}_t$ denotes the filtration up to time $t$.
	\end{lemma}
\begin{proof}
The proof idea is as follows. 
Since the graph is smooth, there are at most $|S_t|/3$
nodes with a large number of edges to $S_t$.
In the remainder of the proof we will show that w.h.p. all other nodes will set their label correctly (including most of the nodes of $S_t$ itself).
This implies, by taking Union bound over all nodes, that w.h.p. $|S_{t+1}|\leq |S_t|/3$.

Fix an arbitrary node $u$ with $Y_u^{S}=0$. 
The number of edges $u$ has within its own cluster
is at least $np/2-\sqrt{12 n p\log n}  $ w.p. at least $1-n^2 $,
by Chernoff bounds with parameter $\delta=\sqrt{\frac{12 \log n}{n p}}$.

Again, let $X_u$ ($X'_u)$, respectively) denote the number of times an edge $(u,v)$ was chosen  where $v$'s current label is the same as $u$'s ground-truth label (opposite, respectively).
We get, by \autoref{lem:smooth} that

\[ \Pro{X_u}=  \frac{ np/2-\sqrt{12 n p\log n}  -\Delta n/12 }{|E|} .\]

We assume w.l.o.g. that $q =\Omega(\log n/n)$ otherwise, the proof follows trivially since \eqref{eq:assumptions} implies that $p =\Omega(\log n/n)$.

\[ \Pro{X'_u}=  \frac{ nq/2+\sqrt{12 n q\log n}  +\Delta n/12 }{|E|} \]

Note that $n^2 p/8 \leq |E| \leq 2 n^2 p$ w.h.p.
Applying the same argument as in \autoref{lem:onephasepq} with $p_1 = 
\Pro{X_u}$ and $p_2 = \Pro{X'_u}$,
we get that $r \geq \frac{72 pn \log n  }{ (\sqrt{p''}-\sqrt{q''})^2  }$ rounds are sufficient w.h.p.

For completeness we give the proof.
W.l.o.g. $q=\Omega(\log n/n)$.
Let $\delta=\sqrt{6 \log n/\Ex{X_u}} , \delta'=\sqrt{6 \log n/\Ex{X'_u}}$ and observe that $\delta,\delta'\leq 1$.
By Chernoff bounds, we get,

\begin{align}
	\Pro{X_u \geq (1-\delta)\Ex{X_u} } + \Pro{X'_u \leq (1+\delta')\Ex{X'_u} }   
	\leq 2 e^{-2 \log n}.
\end{align}

Conditioning on this, we have
\begin{align}
X_u -X_u'\geq (1-\delta)\Ex{X_u}- (1+\delta)\Ex{X'_u}  = r p_1-\sqrt{6 \log nrp_1} - 
\left( r p_2 +\sqrt{6\log n rp_2} \right) > 0,
\end{align}
where the last inequality holds as long as  $r \geq \frac{6\log n}{ (\sqrt{p_1}-\sqrt{p_2})^2 }$.

We have \begin{align*}
 r=\frac{72p n \log n  }{ (\sqrt{p''}-\sqrt{q''})^2  } > 
 \frac{6\log n}{\frac{n}{|E|}(\sqrt{p''}-\sqrt{q''})^2}
 \geq \frac{6\log n}{(\sqrt{p_1}-\sqrt{p_2})^2}.	
 \end{align*}

%
%
%

We remark that  if $q=p/2$, then number of rounds becomes
$r=O(n\log n)$.
First, observe that \eqref{eq:assumptions} implies that $p \geq 2^{12}\log n/n$ and
  $\Delta \leq p/2 $. Hence, $\sqrt{p''}-\sqrt{q''} \geq  \sqrt{\frac{p}{2}-p\sqrt{12/2^{12}}  -p/12 } -  \sqrt{\frac{p}{4} +\frac{p}{2}\sqrt{12/2^{12}}  +p/24 }\geq \frac{\sqrt{p}}{2^{11}}$ and hence
 $r=O(n \log n).$

This completes the proof.

\end{proof}

\begin{proof}[Proof of \autoref{thm:cleanupg}]
The proof from \autoref{lem:smooth} together \autoref{lem:exepcteddrop} and applying Union bound over at most $6\log n$ rounds.
\end{proof}

\begin{theorem}[{\cite[Equation 10]{HR90}}]\label{thm:bounds-binomial-distribution}
Let $Y=\sum_{i=1}^m Y_i$ be the sum of $m$ i.i.d.~random variables with
$\Pro{Y_i=1}=p$ and $\Pro{Y_i=0}=1-p$. We have for any
$\alpha \in (0,1)$ that
\begin{equation*}
\Pro{Y \geq \alpha \cdot m } \leq \left(\left(\frac{p}{\alpha} \right)^\alpha \left(\frac{1-p}{1-\alpha} \right)^{1-\alpha}\right)^m \enspace .
\end{equation*}
\end{theorem}

\begin{theorem}[{Chernoff bound \cite[Theorem 4.4 and 4.5]{MU05}}]\label{chernoff}
Let $X=\sum_i X_i$ be the sum of $0/1$ independent random variables.
Then,
\begin{enumerate}
	\item for any $\delta >0$, \[ \Pro{X\geq (1+\delta) \Ex{X}} < \left(\frac{e^\delta}{(1+\delta)^{1+\delta)}} \right)^{\Ex{X}} .\] 
    \item for any $0 < \delta \leq 1$, \[ \Pro{X\geq (1+\delta) \Ex{X}} \leq e^{-\Ex{X}\delta^2/3}. \] 
	\item for $R\geq 6 \Ex{X}$,  \[ \Pro{X \geq R} \leq 2^{-R}. \]
	\item  for $0 < \delta < 1$, \[ \Pr{X\leq (1-\delta) \Ex{X}} \leq \left(\frac{e^{-\delta}}{(1-\delta)^{1-\delta)}} \right)^{\Ex{X}}. \] 
	\item for $0 < \delta < 1$, \[ \Pro{X\leq (1-\delta) \Ex{X}} \leq e^{-\Ex{X}\delta^2/2} . \] 
\end{enumerate}
\end{theorem}

\end{document}